\RequirePackage{fix-cm}
\documentclass[smallextended]{svjour3}       
\smartqed  
\usepackage{graphicx}
\usepackage{xy}
\usepackage[pdftex,table]{xcolor}
\usepackage{proof}
\usepackage{amsmath}
\usepackage{amssymb}
\usepackage{stmaryrd}
\usepackage[bbgreekl]{mathbbol}
\usepackage{esvect}
\newcommand{\commento}[1]{}
\usepackage[braket]{qcircuit}
\usepackage{marginnote}
\usepackage[normalem]{ulem}

\newcommand{\MZ}[2]{\textcolor{gray}{\sout{#1}}\color{magenta}\  #2 \color{black}}

\usepackage{todonotes}
\presetkeys{todonotes}{size=\scriptsize,linecolor=blue,bordercolor=red,backgroundcolor=yellow!10}{}
\newcommand{\todoReverse}[2][]{{%
 \let\marginpar\marginnote
 \reversemarginpar
 \renewcommand{\baselinestretch}{0.8}%
 \todo[#1]{#2}}}
\newcommand{\todoRight}[2][]{{%
 \let\marginpar\marginnote
 \reversemarginpar
 \renewcommand{\baselinestretch}{0.8}%
 \todo[#1]{#2}}}

\newcommand{\todoAUTO}[2][]{
\ifodd\value{page}
\renewcommand{\baselinestretch}{0.8} \todo[#1]{#2}
\else
\let\marginpar\marginnote\renewcommand{\baselinestretch}{0.8}\reversemarginpar \todo[#1]{#2}
\fi }

\newtheorem{notation}{Notation}

\newcommand*{\myDef}{\mathrel{\vcenter{\baselineskip0.5ex \lineskiplimit0pt
                     \hbox{\scriptsize.}\hbox{\scriptsize.}}}%
                     \mathrel{\vcenter{\baselineskip0.5ex \lineskiplimit0pt
                     \hbox{\scriptsize.}\hbox{\scriptsize.}}}%
                      =
                     }

\DeclareFontEncoding{LS1}{}{}
\DeclareFontSubstitution{LS1}{stix}{m}{n}
\DeclareSymbolFont{symbols2}{LS1}{stixfrak}{m}{n}
\DeclareMathSymbol{\typecolon}{\mathbin}{symbols2}{"25}

\newcommand{\PRB}[1]{\!{#1}\!}

\newcommand\surfCirc{\overline{\mathrm{iX}}}
\newcommand\wfi[1]{{\surfCirc(#1)}}
\newcommand\dom{\textrm{dom}}
\newcommand\codom{\textrm{ran}}

\newcommand{\CIRC}[1]{\mathop{\textsf{circ}\mkern-1mu(\mkern-1mu#1\mkern-1mu)\mkern-1mu}}

\newcommand{\Nat}{\textsf{Nat}}
\newcommand{\Idx}{\textsf{Idx}}

\newcommand{\revCirc}{\mathop{\texttt{reverse}}}
\newcommand{\appCirc}{\mathop{\typecolon}}
\newcommand{\iterCirc}{\mathop{\texttt{iter}}}
\newcommand{\dMeas}{\mathop{\tt dMeas}}
\newcommand{\HEval}{\text{Hilb}}

\newcommand{\qpcf}{\text{\textsf{qPCF}}}

\def\TTY{{\tt Y}}

\newcommand{\condinc}[2]{\ifthenelse{\equal{\typeof}{0}}{#1}{#2}}

\newcommand{\restr}[1]{{\mathop{\upharpoonright_{\! #1}}}}

\def\evaluates{\ev}

\def\fleche{\rightarrow}

\newcommand{\ev}{\mathrel{\pmb\Downarrow}}

\newcommand{\bitExtraction}[2]{\lceil #1 \rceil^{#2}}

\newcommand{\ifz}{\texttt{\,if}}
\newcommand{\lif}[3]{\ifz\,\, {\tt #1 }\,\, {\tt #2} \,\, {\tt #3}\,}

\def\PCF{\text{\textsf{PCF}}}

\newcommand{\FV}{\mathrm{FV}}

\newcommand\num[1]{{\tt \underline{#1}}}

\newcommand{\n}{\num{n}}
\newcommand\equal[2]{{\tt eq}_{#1}({\tt #2})}

\def\Pred{{\tt pred}}
\def\Succ{{\tt succ}}

\newcommand{\size}{\texttt{size}\,}
\newcommand{\ifIdx}{\texttt{if}^\texttt{x}\,}
\newcommand\dotminus{\mathbin{\mathchoice{\scriptstyle\dodotminus\displaystyle}{\scriptstyle\dodotminus\textstyle}{\scriptscriptstyle\dodotminus\scriptstyle}{\scriptscriptstyle\dodotminus\scriptscriptstyle}}}
\newcommand\dodotminus[1]{\dot{\smash{#1-}}}
\newcommand{\getBit}{\texttt{get}}
\newcommand{\setBit}{\texttt{set}}

\newcommand{\redto}{\rightarrow}

\newenvironment{myArray}[1][1]{%
  \renewcommand*{\arraystretch}{#1}%
  \array%
}{%
  \endarray
}

\usepackage{url}

\begin{document}

\title{{QPCF}: higher order languages and quantum circuits 
}

\titlerunning{qPCF}        

\author{\small Luca Paolini \and Mauro Piccolo \and Margherita~Zorzi 
}


\date{Received:  2017 / Accepted: 2018 / Published online: \today}

{
\institute{L. Paolini  \and M. Piccolo \at
              Department of Computer Science, University of Torino (Italy)\\
              \email{paolini@di.unito.it,piccolo@di.unito.it}           
           \and
           M. Zorzi \at
              Department of Computer Science, University of Verona (Italy)\\
               \email{margherita.zorzi@univr.it}   
}
}
\maketitle

\begin{abstract}
\commento{
qPCF is a paradigmatic programming language that provides flexible facilities to manipulate quantum circuits by restricting quantum measurements to total ones.
qPCF follows the tradition of ``quantum data \& classical control'' languages, and rests on a QRAM model, where no intermediate quantum state is permanently stored
and the emphasis is moved from states to circuits.
We introduce qPCF syntax, typing rules, and its operational semantics. We prove fundamental  properties of the system, such as Preservation and Progress Theorems. 
Moreover, we provide some higher-order examples of circuit encoding.
}{
  qPCF is a paradigmatic quantum programming language that extends PCF with quantum circuits and a quantum co-processor.
  Quantum circuits are treated as  classical data that can be duplicated and manipulated in flexible ways by means of a dependent type system.
  The co-processor is essentially a standard QRAM device, albeit we avoid to store permanently quantum states in  between two co-processor's calls.
  Despite its quantum features, qPCF retains the classic programming approach of PCF.
  We introduce qPCF syntax, typing rules, and its operational semantics. We prove fundamental  properties of the system, such as Preservation and Progress Theorems. 
Moreover, we provide some higher-order examples of circuit encoding.
  }
\keywords{PCF \and Quantum Computing \and Quantum Programming Languages}
\end{abstract}

\section{Introduction}\label{sec:introduction}
Quantum computing is an intriguing trend in computer science research. 
The interest about quantum computing is due to R. Feynman. In~\cite{Be80}, the first relationships between the quantum mechanical model and a formal computational model such as Turing machines is stated. The first concrete proposal for a quantum abstract computer is due to D. Deutsch, who introduced quantum Turing machines~\cite{Deu85}.
Later, a number of results about
computational complexity~\cite{Shor94,Shor99,Gr99,JUW11} showed quantum computing is not only a  challenging theoretical subject, but also a promising paradigm for the concrete realization of powerful machines.\\ 
Nowadays, quantum computers are a long term {industrial goal and a} tangible reality {in term of prototypes}. Even if physicists and engineers have to face tricky problems in the realization of quantum devices, the advance of these innovative technologies {is ceaseless}. 
As a consequence, to fully understand how to program quantum devices is become an urgent need. 

Typically, calculi for quantum computable functions present two  different  computational {features}. 
On the first hand, there is the unitary aspect of the calculus, that captures the essence of quantum computing as  algebraic transformations. 
On the other hand, {it should be possible to} \emph{control} the quantum steps, 
``embedding'' the pure quantum evolution in a classical computation.
 Behind this second point, we have the usual idea of computation as a sequence of discrete steps on (the mathematical description of) an abstract machine. 
The relationship between these different aspects gives rise to different approaches to quantum functional calculi (as observed in~\cite{ArrDow08}). 
If we divide the two features, i.e. we separate data from control, we adopt the so called \emph{quantum data \& classical control} ($qd\&cc$) approach.
This means that {quantum} computation 
is \emph{hierarchicallly dependent} from the {classical} part: a classical program (ideally in execution on a classical machine) computes some ``directives" and, these directives are sent to a hypothetical device which applies them to quantum data. Therefore quantum data are manipulated by the classical program {or, in other words,} classical computational steps  control the unitary part of the calculus.

This idea is inspired to an architectural model called Quantum Random Access Machine (QRAM). The QRAM has been defined in~\cite{Knill96} and can be viewed as a  classically controlled machine enriched with a quantum device. 
On the grounds of the QRAM model, P. Selinger defined the first functional language based on the  {quantum data-classical control} paradigm~\cite{selinger2006mscs}. 
This work represents a milestone in the development of quantum functional calculi and inspired a number of different investigations.
Many  quantum programming languages implementing the $qd\&cc$ approach can be found in literature~\cite{selinger2006mscs,selinger09chap,qwire,Ross2015}.
Details about related works are in Section \ref{sec:related}.

In this paper we propose some new contributions to the research on the $qd\&cc$ paradigm by formalising the  quantum language 
$\qpcf$, based on a simplified version of the QRAM model restricted to total measurements (cf. Section \ref{sec:qram}).
$\qpcf$ extends \PCF, namely {the prototype of} typed functional language. {Some $\qpcf$ features is listed below.}
\begin{itemize}
\item \emph{Absence of explicit linear typing constraints}: 
the management of linear resources is radically different from the mainstream in 
{languages inspired to Linear Logic} such as~\cite{selinger2006mscs,dallago2009mscs,DLMZ10tcs,DLMZ11entcs,Zorzi16,Ross2015,Felty16}; so, we do not use linear/exponential modalities.
\item \emph{Use of dependent types} : 
we decouple the classical control from the quantum computation \commento{by adopting a simplified form of}{and we use} dependent types \commento{that provide a sort of  linear  interface}{to manage quantum circuits}.
 Dependent types  allow to write not only individual quantum circuits, but also parametric programs.
\item \emph{Emphasis on the Circuit Construction}: 
  \commento{ We focus on the perspective  ``$\qpcf$ as quantum circuit description language'', following recent trends in quantum programming theory~\cite{quipperacm,qwire}.
    We pursue this idea because it is also in accord with the design practice of quantum algorithms.
  }
  {
    In accord with the recent trends in quantum programming theory~\cite{quipperacm,qwire,rios2017eptcs}, $\qpcf$ focuses on the quantum circuit description aspects.
This idea already considered in~\cite{Knill96} aims to ease the programming of quantum algorithms.
  }
\item 
\commento{\emph{No explicit management of quantum states/Total measurements}. Differently from other proposals where an explicit notion of (quantum)
program state is introduced (e.g.~\cite{selinger2006mscs,dallago2009mscs,Zorzi16}), $\qpcf$ avoids the need of a type for quantum states (and its linear management). 
$\qpcf$ hides the quantum evaluation and measurement into {the operator} $\mathsf{dmeas}$ (cf. Section~\ref{sec:eval}), that both evaluates programs representing circuits and implements the \emph{(von Neumann) Total Measurement}~\cite{NieCh10}. Thanks to the deferred measurement principle~~\cite{NieCh10}, this restriction does not represent a theoretical limitation, as discussed in Section~\ref{sec:discussion}.
}{
  \emph{No permanent quantum states}.
  Differently from other proposals (e.g.~\cite{selinger2006mscs,dallago2009mscs,Zorzi16}), $\qpcf$ does not need types for quantum states (and its linear management).
  This is possible, since the interaction with the quantum co-processor is neatly decoupled by means of the operator $\mathsf{dmeas}$.
  It offloads a quantum circuit to a co-processor for the evaluation which is immediately followed by a \emph{(von Neumann) Total Measurement}~\cite{NieCh10}.
  This means that partial measures are forbidden. Luckily, the deferred measurement principle~~\cite{NieCh10} says us that this restriction does not represent a theoretical limitation
  (cf. Section~\ref{sec:discussion}).
}
\end{itemize}


  $\qpcf$ is an higher-order programming language as \PCF\ that retains the standard classic programming approach.
  Potentially, this can ease the transition to computers endowed with quantum co-processors.

A preliminary version of $\qpcf$ has been proposed in~\cite{lmtamc}. With respect to the first version, we extended our proposal following several directions. 
\begin{itemize}
\item We slightly reformulate syntax and heavily improve the type system, that is radically more refined. 
\item We carefully provide proof's details of theoretical properties of $\qpcf$, only stated in the first version.
  Proofs are non trivial, since one has to  consider infinite ground types and to unravel the mutual relationships that holds between syntactic classes.
  \item {We extende the semantic in order to take into account the probabilities of quantum measurements.}
\item We add many examples. 
\item We provide a deep discussion about related work by focusing on classical control quantum languages. 
  \item We define a restricted version of the general QRAM (cf. Section \ref{sectionQRAM}) that represents the architecture required by $\qpcf$.
\end{itemize} 

Summing up, we propose a new, {stand-alone} quantum programming language that aims to  
\commento{couple simplicity and higher order features, allowing parametric programming.
  Differently from $\mathsf{QWIRE}$, that shares with $\qpcf$ some ideas (see Section~\ref{sec:related}), we follows the functional quantum calculi mainstream,
  and we provide a self-contained language, since it does not rely on external language and type theory.Focusing on total measurements and on 
  a restricted version of the QRAM architecture, we aim to propose the (as much as possible) simpler extension of \textsf{PCF} able to tame the quantum circuit generation and evaluation.
}{
combine classical programming style, parametric circuit programming, higher order and quantum features in a unified  setting.
}

\subsubsection*{Synopsis}
In Section~\ref{sectionQRAM} we describe the ideal quantum programming environment
 (a restricted version of the QRAM machine) behind $\qpcf$. 
Section~\ref{sec:qpcf} and Section~\ref{sec:typing} introduce the syntax and typing rules of $\qpcf$ respectively.
The operational semantic of $\qpcf$  and main properties of the system are in Section~\ref{sec:eval}.
In Section~\ref{sec:examples} we 
{
  discuss the implementation of some quantum algorithms.}
 A detailed overview about the state of art of classical control quantum languages is in Section~\ref{sec:related}.
In Section~\ref{sec:discussion} we propose some conclusive considerations with special care for the language expressivity. 



\section{Background:  the QRAM model}\label{sectionQRAM}

{We assume some familiarity with notions as \emph{quantum bits} (or \emph{qubits} the quantum equivalent of classical data),  
quantum states~\cite{NieCh10} (systems of $n$ quantum bits), quantum circuit and quantum circuit families~\cite{NishOz09}. 
} 
In Section~\ref{sec:qram} we introduce a simplified version of the QRAM (the architectural model behind the quantum data and classical control approach) called rQRAM.   
The rQRAM represents the ideal hardware  to execute $\qpcf$.
 

\subsection{Introducing an idealized co-processor}\label{sec:qram}
 
The physical realization of  basic components necessary for universal quantum computation has gathered much attention in recent years,
 and many realistic technologies have emerged (see  \cite{metodi2011book} for a survey): trapped ions, quantum dots, polar molecules, and superconductors among others.
It is commonly accepted that physical quantum computer has to fulfill the following five criteria which were proposed by Di~Vincenzo at IBM in \cite{diVincenzo2000}:
(i) a scalable physical system of qubits; (ii) the ability to initialize the state of the qubits;
(iii) long relevant decoherence times (qubits lose their quantum properties exponentially quickly);
(iv) a universal set of quantum gates; and (v) a measurement capability.
These criteria correspond quite directly to various engineering hurdles that implementations have to face.
In particular, in the last years many efforts have been supplied to increase the decoherence times and 
to overcome  adjacency/neighboring constraints on qubits  (see  \cite{bettelli2003tepj,diVincenzo2000,metodi2011book}).

{
It is clear that future quantum hardware may differ in many details, so we have to look for some abstract model of quantum computations.
$\qpcf$ is designed to be executed on the QRAM programming environment (see \cite{Knill96,selinger2004mscs,miszczak2014book}) which is commonly accepted as a reasonable model of computation
for describing quantum computing devices.
  }
However, $\qpcf$ rests on a restricted QRAM model that aim to lighten two problems:
(i) decoherence care is needed only during the evaluation of a single circuit and,
(ii) free-rewiring is not required, since we strictly rest on the basic gates provided by the co-processor.

We quote the QRAM introduction provided in \cite{Knill96}. 
\begin{quote}\small
  It is increasingly clear that practical quantum computing will take  place on a classical machine with access to quantum registers. 
The classical machine performs off-line classical computations and  controls the evolution of the quantum registers by initializing them
  to certain states, operating on them with elementary unitary operations and measuring them when needed.
\end{quote}
We also quote some concise remarks done by Selinger in \cite{selinger2004mscs}.
\begin{quote} \small
  Typically, the quantum device will implement a fixed, finite set of unitary transformations that operate on one or two quantum bits at a time. 
The classical controller communicates with the quantum device by sending a sequence of instructions, specifying which fundamental operations are to be performed. 
The only output from the quantum device consists of the results of measurements, which are sent back to the classical controller.
\end{quote}
It is worthwhile to note that the QRAM model is sometimes considered too restrictive to support proposed quantum programming languages;
therefore, it is enhanced to include more possible interactions between the classical and the quantum devices.
\commento{:
in particular, it is allowed to quantum device to return  ``instructions''. We quote from \cite{valiron2013ngc}[p.3].
\begin{quote}\small 
The device holds the quantum memory and can allocate new quantum bits and measure or apply unitary maps on existing quantum bits at wish. 
The quantum device is a black-box whose entries are instructions and whose only outputs are the results of measurements and return codes of instructions.
\end{quote}
}{(see \cite{valiron2013ngc}[p.3])}

{
  $\qpcf$ is designed to operate in a restricted QRAM programming environment} named rQRAM, that corresponds, quite well, to an ideal  co-processor for a classical computer.
The idea is that our classical computers compute  circuits (i.e. a sequences of gates) that are classical data.
A circuit can be offloaded to the quantum device in order to be applied to a quantum register suitably initialized.  
The co-processor has to be able: (i) to  initialize a register to  a given classic value; 
(ii) to apply  a given sequence of gates on the state stored in the register;
(iii) to perform a final quantum measurements of the whole register.
It is worthwhile to remark some peculiarities of our rQRAM. It neatly isolates the application of non-unitary transformations from unitary ones;
the set of available quantum gates (hopefully,  a universal set of quantum gates) determines the possible permutations of qubits; and, more interestingly, 
the register is not assumed to be able to store permanently qubits, overcoming unpredictable  decoherence times.  
Indeed, the rQRAM co-processor has to support a bounded decoherence time (the maximal of times needed to the application of gates in the device)
and it is even possible to imagine how to
include several classical controllers to share the access to the same single quantum device. 
However, $\qpcf$ can be easily executed on standard QRAM models.  
We are convinced that its realization can be easier that the standard QRAM models and interesting for implementation.
In Section~\ref{sec:expressivepower} we explain why our model does not cause any theoretical limitation,
albeit some algorithms cannot be directly executed on it.
\section{\qpcf}\label{sec:qpcf}

   $\qpcf$ is quantum programming language
   based on the lambda-calculus, like many other quantum calculi 
     (see, for instance, ~\cite{selinger2006mscs,dallago2009mscs,Zorzi16}).
   In accord with the recent trend in the development of quantum programming languages \cite{green2013acm,lmtamc,qwire,rios2017eptcs,Ross2015},
   $\qpcf$ focuses on the abilities to generate and manipulate quantum circuits. 

  No  permanent quantum state can be stored in $\qpcf$, thus linear types are avoided:
  this distinguishes $\qpcf$ from  other typed quantum programming languages.
  The linearity for quantum control is completely confined to atomic datatypes by using a simplified form of dependent types as that suggested in \cite{pierce2002mit}:
  a dependent type picks up a family of types that bring in the type auxiliary information (just the arity of a circuit, in our case).

\subsection{Syntax Overview.}\label{subSect:overview}

$\qpcf$ extends $\PCF$ \cite{plotkin77tcs,gaboardi2016} to manage some additional atomic data structures: \emph{indexes} (normalizing number expressions)  and \emph{circuits}.
Index expressions are essentially built by means of variables, numerals and some \textit{total} operations on expressions: $\odot\in\{ +, * \}$ (viz. sum, product). 
Circuits expressions are obtained by means of  suitable operators combining gates; their evaluation is expected to produce (when terminating) strings on gates. 

The row syntax of $\qpcf$ follows:
$$
\begin{array}{rrl}
{\tt M,N,P,Q} &\myDef&  \texttt{x} \mid {\tt \lambda x.M} \mid {\tt MN} \mid \num{n} \mid \Pred\mid \Succ\mid \ifz \mid \TTY_\sigma \mid \setBit \mid \getBit\\
& \mid &    \texttt{U} \mid  \appCirc \mid\;\parallel\; \mid \iterCirc  \mid \revCirc  \mid \odot \texttt{E} \texttt{E}' \mid \size \mid \dMeas 
\end{array} 
$$
\noindent where $\texttt{E}$ ranges over index expressions, namely terms typed as indexes.

 In the first row, we include $\PCF$ extended with some syntactic sugar in order to facilitate the bitwise access to numerals:
 $\getBit$ allows us to extract (to read) the $i$-th digit of the binary representation of a numeral, i.e. its $i$-th bit; 
$\setBit$ allows to modify the $i$-th bit of a numeral. They are added to simplify the initialization and decomposition of states.

$\qpcf$ is  parameterized by a set of quantum gates that correspond to the unitary operators made available by the quantum co-processor. 
We  assume $\texttt{U}$ to range on available gates and,
in accord with \cite{qwire}, we can assume that a universal subset of unitary gates (see \cite{KLM07,Kitaev2002,NieCh10}) is available.
Let $\mathcal{U}$ be the set of  available computable unitary operators, 
such for each gate  ${\tt U}$ there is a unique $\mathbf{U}\in\mathcal{U}$.
If $k\in\mathbb N$ then we denote $\mathcal{U}(k)$ the gates in  $\mathcal{U}$ having arity $k+1$,
so $\mathcal{U}=\bigcup_0^\omega \mathcal{U}(k)$. More explicitly, the gates of arity $1$ are in $\mathcal{U}(0)$ and so on.

The syntax of (evaluated) \emph{circuits} is generated by 
$\typecolon,\parallel $ and gate-names. The symbol
$\appCirc$ sequentializes two circuits of the same arity, while $\parallel$ denotes the parallel composition of circuits.
We build (to be evaluated) circuit expressions by means of $\iterCirc$ and $\revCirc$.
We use $\iterCirc$ to produce the parallel composition of a first circuit with a given numbers of a second one. 
We use $\revCirc$  to transform a circuit in another one of the same arity.

Indexes are operated by  \textit{total} operations. W.l.o.g. we assume $\odot\in\{ +, * \}$ (viz. sum, product). 
Types of circuits include index expressions.
We add $\size$ to $\qpcf$ to emphasize the gain that dependent type can concretely provide, although this makes the proofs of the language properties more complex.
$\size$ is an operator that applied to a circuit-expression returns the arity of the corresponding circuit, viz. an index information. 

Last, but not least, we use $\dMeas$  to evaluate circuits suitably initialized:
$\dMeas$ returns a numeral being the binary representation of a quantum measure executed on all quantum wires (of the circuit) as a whole.

$\qpcf$ is indeed conceived to manage circuits that can be freely duplicated and erased, 
while quantum states are hidden by means of $\dMeas$. In some sense,
$\dMeas$ offloads a quantum circuit to a co-processor,
it waits  the end of circuit execution and it returns the final measure of all wires.

\subsection{Dependent Types.}

Dependent types are widely used in proof-theoretical research. Typically, in presence of strongly normalizing languages. 
Unfortunately, the type-checking of dependent types requires to decide the equality of terms (that can be included in types)
and the strong normalization is not realistic for programming languages~\cite{luca06}. 
Therefore, the management of terms in types is the crucial issue that has to be faced in a programming language. 
 This point is discussed in page 75, Section 2.8 of \cite{aspinall05mit} where  different programming approaches are compared.
 Following  \cite{xi1999acm,zenger1997tcs}, $\qpcf$ forbids the inclusion of arbitrary terms in types. 
A suitable subclass of terms (numeric expressions always normalizing) is identified to this purpose: these terms are called \emph{indexes}.
 
Our approach to dependent types is closely inspired to that mentioned in \cite[\S 30.5]{pierce2002mit} to manage vector's types:
the decoration carries with it, some dimensional (i.e. numeric) information.
We avoide general dependent types systems (see \cite{aspinall05mit} for a survey)  because their great expressiveness is exceeding our needs.
We prefer to maintain {the $\qpcf$ type system} as simple as possible
in order to show the feasibility of the approach and its concrete benefits. 

\subsection{Types Overview.}

Traditional types of $\PCF$ (i.e. integers and arrows) are extended to include 2 new types:
a  type for \emph{indexes} (strong normalizing numeric expressions) and a type for \emph{circuits} (carrying around indexes).
 Types of $\qpcf$ are formalized as follows:

 $$\sigma,\tau \myDef \Nat \mid\Idx \mid    \CIRC{\texttt{\texttt{E}}}\mid \Pi\texttt{x}^\sigma.\tau$$

\noindent where $\texttt{E}$ is an index expression (viz. a normalizing numeric expression).
As common in presence of dependent types, we replace arrows by quantified types that include more information than arrows: 
they make explicit, in the type, the variable-name which is bounded.
The variable-name can be $\alpha$-renamed whenever the standard capture-free proviso is satisfied.
We use $\sigma \fleche \tau$ as an abbreviation for $\Pi\texttt{x}^\sigma.\tau$ whenever,
either the variable-name $\texttt{x}$  does not occur in $\tau$ or we are not concerned with $\texttt{x}$.

\subsection{Indexes.} 
The type $\Idx$ picks up a subset of numeric normalizing expressions. 
In particular,  $\Idx$  does not allow to type  non-normalizing expressions as  $\TTY (\lambda \texttt{x}.\texttt{x})$. 
The main use of $\Idx$ is to type all proper-terms used in dependent types.
Indexes are built on numerals and total operations on them.
We limit our operations to addition and multiplication,
 but we assume that this set is conveniently tuned in a concrete case,
e.g. by adding the (positive subtraction) $\dotminus$, or the modulus $\%$, or a selection $\ifIdx$ and so on.
See Remark \ref{rem:typecheck} for more details.

\section{Typing system}\label{sec:typing}

Finite sets of pairs ``variable:type'' are called \emph{bases} whenever variable-names are disjoint: we use $B$ to range over them.
We denote $\dom(B)$ the finite set of variable names included in $B$ and we denote $\codom(B)$ the finite set of types that $B$ associates to variables.
As usual, we use the notation $B\cup\{{\tt x}:\sigma\}$ to extend the base $B$ with the pair ${\tt x}:\sigma$ where
 is assumed that we are adding a fresh variable, i.e. $\mathtt{x}\not\in \dom(B)$.
We assume the analogous disjointness conditions writing $B\cup B'$, for any $B'$.

$\qpcf$ includes a typing axiom of the shape $x:\sigma \vdash x:\sigma$ for each type $\sigma$.
In presence of dependent types, $\sigma$ can contain terms and, consequently,
it is a valid type only when such terms are well-typed. 
Luckily we have a unique type that includes dependencies, namely the type of circuit; moreover, it includes a term that must be typed as index. 
We solve this circularity by means of a function extracting all terms included in types.
More precisely, the function has to identify a set of statements (a.k.a. typings) of the shape $B\vdash \mathtt{E}:\Idx$ where $B$ are bases,
that are needed to check the validity of the considered type.

\begin{definition}
The function $\surfCirc$ takes in input a base and a type and, it returns  a set of typings.
We define $\surfCirc$  by induction on the second argument, as follows:
$\surfCirc(B,\Idx):= \emptyset$,  $ \surfCirc(B,\Nat):= \emptyset$, $\surfCirc(B,\CIRC{\texttt{E}}):=\{ B\vdash \texttt{E}:\Idx\}$
  are the base cases, and 
  $\surfCirc(B, \Pi{\tt x}^\sigma.\tau) := \surfCirc(B, \sigma)\cup \surfCirc(B\cup  \{\mathtt{x}:\sigma\}, \tau)$ is the inductive one,
where w.l.o.g. we are assuming that $\mathtt{x}\not\in \dom(B)$. 
By abusing the notation, we extend $\surfCirc$  to consider set of types, 
namely $\surfCirc(B,S):=\bigcup_{\tau \in S} \surfCirc(B,\tau)$  where $S$ is a finite set of types.
\end{definition}

\begin{example}\label{example:ONE}
  $\surfCirc(\{\texttt{x}:\Idx \}, \texttt{z}: \CIRC{ \texttt{z}\oplus\texttt{E}' })$
  identifies  $\{ \texttt{x} :\Idx \vdash \texttt{z}\oplus\texttt{E}':\Idx\}$. 
   It is worth to note that $\texttt{E}'$ can (hereditarily) contain subterms typed as circuits.\qed
\end{example}

Usually, dependent types are formalized via the introduction of super-types (named kinds) and super-typing rules (a.k.a. kinding rules)
that identify well-given types, see \cite{aspinall05mit} for instance. To limit the number of kinding rules, 
dependent type systems introduce formal tricks that allow to re-use the (ground) typing rules in the kinding system.
We taken advantage of our limited use of dependent types to circumvent the introduction of kinds and kinding rules aiming at:
(i) to avoid the explosion of the number of rules and complex overlapping of rules;
(ii) to make the extension of $\PCF$ as clear as possible; 
and, (iii) to avoid further complexity (viz. mutual induction) in proofs.         
For the sake of completeness, we remark that it is possible to reformulate the use of $\surfCirc$ by adding a unique kind $\square$ to identify 
well-given types, and to add kinding rules checking that terms in types are well-typed.

\begin{definition}\label{defTypeRules}
  The rules of the typing system are  given in Table \ref{TypingRules}.
A typing is \emph{valid} whenever it is the conclusion of a finite type derivation built on the given rules.
We use $\surfCirc$ applied to some arguments in some rule's premises as a placeholder for the resulting typings.
\noindent  Types are considered up to the congruence $\simeq$, which is the smallest equivalence including:
 (i) the $\alpha$-conversion of bound variables;
 (ii) $\beta$-inter-convertibility of $\beta$-redexes (occurring in terms included in types);
 (iii) associativity, commutativity and distributivity of sum and product together with the properties of neutral elements (viz. $0$, $1$).
\end{definition}

We consider types up to inter-convertibility of included terms. 
For the sake of simplicity, we avoid to formalize the interconvertibility via additional rules.
Moreover, we remark that we introduce the untyped row syntax (cf. Section \ref{sec:qpcf}) only in order to help the reading.
Nevertheless, we are interest only to explicitly typed terms (\`a la Church) where all terms are considered together with their whole typing information.

\begin{table}[t]
$$\begin{myArray}[1.8]{c}
     \infer[\scriptstyle(P_0)]{   B \cup  \{\mathtt{x}:\sigma\} \vdash {\tt x}:\sigma }{ \wfi{B,\codom(B)\cup  \{\sigma\}}}
\quad \infer[\scriptstyle(P_1)]{B\vdash {\tt \lambda x^{\sigma}.N}:\Pi{\tt x}^{\sigma}.\tau}{B\cup  \{ {\tt x}:\sigma\}\vdash {\tt N}:\tau }
     \quad  \infer[\scriptstyle(P_2)]{B\vdash {\tt PQ}:\tau[\texttt{Q}/\texttt{x}]}{B\vdash {\tt P}: \Pi \texttt{x}^{\sigma}.\tau\quad B\vdash \texttt{Q}:{\sigma}} 
\\
\infer[\scriptstyle(P_3)]{B\vdash {\tt succ}:\Nat\fleche \Nat}{\wfi{B,\codom(B)}} 
\qquad \infer[\scriptstyle(P_4)]{B\vdash {\tt pred}:\Nat\fleche \Nat}{\wfi{B,\codom(B)}} 
\\
  \infer[\scriptstyle(P_5)]{B\vdash \ifz:\Nat\fleche \Nat\fleche\Nat\fleche\Nat}{\wfi{B,\codom(B)} } 
  \qquad
  \infer[\scriptstyle(P'_5)]{B\vdash \ifz:\Nat\fleche \CIRC{\texttt{E}}\fleche\CIRC{\texttt{E}}\fleche\CIRC{\texttt{E}}}{   B\vdash \texttt{E}:\Idx  }  \\
\infer[\scriptstyle(P_6)]{B\vdash \TTY_\sigma:(\sigma \fleche \sigma) \fleche \sigma}
   {\sigma=\tau_1 \fleche \ldots \fleche \tau_n \fleche \gamma & {\scriptstyle (\gamma \in \{ \Nat, \CIRC{\texttt{E}} \} \text{ and } n\geq 0)} & \quad \wfi{B,\codom(B)\cup\{\sigma  \}} }\\
\infer[\scriptstyle(B_1)]{B\vdash \getBit: \Nat \fleche\Nat\fleche \Nat}{\wfi{B,\codom(B)}}  \quad
\infer[\scriptstyle(B_2)]{B\vdash \setBit:\Nat\fleche \Nat\fleche \Nat}{\wfi{B,\codom(B)}}  \\
\infer[\scriptstyle(I_0)]{B\vdash {\tt M}:\Nat}{B\vdash {\tt M}:\Idx}
\quad \infer[\scriptstyle(I_1)]{B\vdash \num{n}:\Idx}{\wfi{B,\codom(B)}}
\quad  \infer[\scriptstyle(I_2)]{B\vdash \odot \,\texttt{E}_0\,\texttt{E}_1: \Idx}{B\vdash \texttt{E}_0:\Idx \quad B\vdash \texttt{E}_1:\Idx  } 
      \quad \infer[\scriptstyle(I_3)]{B\vdash \size\,\texttt{M}: \gamma}{B\vdash \texttt{M}:\CIRC{\texttt{E}}}\\ 
\infer[\!\!\scriptstyle(C_1)]{B\vdash {\tt U}:\CIRC{\num{k}}}{ {\tt U}\in\mathcal{U}(k)\overset{\phantom{x}}{\phantom{f}} \quad \wfi{B,\codom(B)} }
    \quad \infer[\scriptstyle(C_2)]{B\vdash \appCirc:\CIRC{\texttt{E}}\fleche\CIRC{\texttt{E}}\fleche\CIRC{\texttt{E}} }{  B\vdash \texttt{E}:\Idx }\\
\infer[\scriptstyle(C_3)]{B\vdash \parallel :\CIRC{\texttt{E}_0}\fleche\CIRC{\texttt{E}_1}\fleche\CIRC{\texttt{E}_0+\texttt{E}_1+1} }{ B\vdash \texttt{E}_0:\Idx \qquad B\vdash \texttt{E}_1:\Idx }
   \quad \infer[\scriptstyle(C_4)]{B\vdash \revCirc:\CIRC{\texttt{E}}\fleche\CIRC{\texttt{E}} }{ B\vdash \texttt{E}:\Idx }  \\
   \infer[\scriptstyle(C _5)]{B \vdash \iterCirc: \Pi{\tt x}^\Idx.\CIRC{\texttt{E}_0}\fleche\CIRC{\texttt{E}_1} \fleche\CIRC{ \texttt{E}_0+((1+\texttt{E}_1)*\texttt{x} ) }}
                              { B\vdash \texttt{E}_0:\Idx \qquad B\vdash \texttt{E}_1:\Idx } \\[1mm] 
  \infer[\scriptstyle(M)]{B\vdash \dMeas:\Nat\fleche\CIRC{\texttt{E}}\fleche\Nat }{B\vdash \texttt{E}:\Idx}\\
\hline
\end{myArray}$$
\caption{Typing Rules.}
\label{TypingRules}
\end{table}

\subsection{Type rules overview} 

Rules $(P_0)$, $(P_1)$, $(P_2)$, $(P_3)$, $(P_4)$, $(P_5)$, $(P'_5)$, $(P_6)$ are directly inherited from $\PCF$. 
We  add $\surfCirc$ constraint in the premises of rules where bases are introduced
in order to ensure that all terms included in (considered) types are well typed. 
Rules $(P_5),(P'_5),(P_6)$ restrict types allowed for conditional and recursive terms
in order to ensure that the evaluation of terms typed with $\Idx$ are strongly normalizing.
Note that the premise of $(P'_5)$ implies that $\wfi{B,\codom(B)}$.
Rules $(P_1),(P_2)$ involve the type binder $\Pi$ that generalizes standard arrow-type and take care of possible free variables in types.
As expected in dependent type systems, $(P_2)$ substitutes the argument also in types. 
In rules $(P_3),(P_4),(P_5),(P'_5),(P_6)$ we use arrows as an abbreviation for $\Pi$-types,
since the corresponding variable-names is never typed with $\Idx$.
Note that all numerals are typed $\Nat$ by means of rules $I_0$ and $I_1$.

\begin{remark}\label{rem:typecheck}
It is worth to notice that the choice of the operators that we admit in index expressions has a strong impact on the decidability of the type-checking of the language
because they can occur in types. First, a set of strong normalizing expression ensures that the evaluation of closed terms can be decided and 
no run-time error can arise. On the other hand, when we build a program we have to manage open terms and open expressions also in types.
More explicitly, let us assume $B\vdash {\tt P}: \Pi \texttt{x}^{\CIRC{\texttt{E}_P}}.\tau$ and $B\vdash \texttt{Q}:{\CIRC{\texttt{E}_Q}}$:
then, we can apply the rule $(P_2)$ only whenever $\texttt{E}_P\simeq \texttt{E}_Q$.
If the language of index expressions is not endowed with a decidable equality then we can run into unwanted programming awkwardness.
The decidability of the identity between index expressions follows from \cite{diCosmo2005lncs}.
Appealing extensions to elementary function or primitive recursive functions are possible, 
but their impact on the programming practice should be carefully considered facing decidability issues~\cite{richardson94}. \qed
 \end{remark}
 The rule $(B_1)$ types $\getBit$, that  returns  $0$ or $1$ when applied to two integers. More precisely it returns the  bit
 (in the binary representation) of the first integer in the position pointed by the second integer. The rule $(B_2)$ types $\setBit$ that  takes in input two integers:
the second one selects a bit in the binary representation of the first one and it returns as output the numeral obtained by setting to $1$ the selected bit in the binary representation of the first numeral.
Their typing agree with these behaviors.

The rule $(I_0)$ allows to use an index expression as a term typed $\Nat$. 
Rules $(I_1),(I_2)$ type our basic index expressions as expected. 
The rule $(I_3)$ brings back in term the arity information included in the type of a circuit. 

Rules $(C_1),(C_2),(C_3),(C_4),(C_5)$ type circuit expressions.
We recall that the index $0$ has to be intended denoting the arity $1$.
$(C_1)$ makes available the basic gates. $(C_2)$ types the sequential composition of circuits having the same arity.
$(C_3)$ types the parallel composition of two circuits.   
$(C_4)$ types an operator that (possibly) transforms a circuit in its adjoint, so the arity is preserved.
$(C_5)$ types the parallel composition of some circuits, namely a base circuit $\texttt{M}$ and some copies of a circuit $\texttt{N}$.

Finally, $(M)$ types an operator taking in input a state (the binary representation of a numeral) and a circuit, that gives back another state.

\begin{example}\label{exampleDependent} 
An interesting example of term that provides evidence of the circularity arising from dependent types follows.\\[-5mm]
$$\infer[\scriptstyle(P_2)]{
{\tt x}: \Pi{\tt z}^\Idx.\CIRC{{\tt z}}\vdash {\tt x}\,\size\!({\tt M}): \CIRC{\size\!({\tt M})}
}{
\infer[\scriptstyle(P_0)]{{\tt x}: \Pi{\tt z}^\Idx.\CIRC{{\tt z}}\vdash {\tt x}: \Pi{\tt z}^\Idx.\CIRC{{\tt z}} }{}\qquad 
\infer[\scriptstyle(I_3)]{\vdash \size(\texttt{M}): \Idx }{\infer*{\vdash{\tt M}:\CIRC{{\tt E}}  }{}}
}$$
    $\mathtt{M}$ can be any closed term of $\qpcf$ typed as circuit
and $\mathtt{E}$ can be any closed term of $\qpcf$ typed $\Idx$. 

Since can be $\mathtt{M}$ can be any closed term, t his example shows that types can (possibly) include sub-terms being either non-terminating or open variables, not typed $\Idx$.
However, such terms are always argument of $\size$ that looks only for the index term in the ``more external'' type: but $\size$ throws away such information
(cf. Definition \ref{evaluationRules} for more details).\qed
\end{example}


\subsection{Some typing properties.} 

Many standard properties can be easily adapted  to typing system in Table~\ref{TypingRules}. 
 If $B\vdash  {\tt M}:\tau$  then  $\FV({\tt M}), \FV(\tau) \subseteq dom(B)$.
 Bases of typing can be weakened, viz. if $B\vdash   {\tt M}:\tau$ and $ dom(B)\cup dom(B')=\emptyset$ then $B\cup B'\vdash  {\tt M}:\tau$.
 Moreover, it is easy to check that $B\vdash  {\tt M}:\tau$ implies that $ \wfi{B,\codom(B)}$.
Straightforward adaptation of Generation Lemmas hold too. 
However  our interest is more focused on  dynamic properties of the typing system  than on its logical properties.

\begin{lemma}[Substitution lemma]\label{substLemma}$\;$\\
   If $B \cup \{ \mathtt{x}: \sigma \} \vdash {\tt M}:\tau $ and $B' \vdash {\tt N}:\sigma$ then $B[{\tt N}/\mathtt{x}]\cup B' \vdash {\tt M}[{\tt N}/\mathtt{x}]:\tau[{\tt N}/\mathtt{x}]$.
\end{lemma}
\begin{proof}
 By induction on the derivation $B\cup \{ \mathtt{x}: \sigma \} \vdash {\tt M}:\tau$. We remark that $\surfCirc$ is a placeholder for a set (possibly empty) of sub-derivations.
\qed
\end{proof}

From Lemma \ref{substLemma} the subject reduction follows easily. 

\begin{lemma}\label{lem:auxExp} 
  Let $\mathcal{D}$ be a derivation concluding $B \vdash {\tt M}[\texttt{N}/\texttt{x}] :\sigma$ where $\texttt{x}$ is fresh.
  \begin{enumerate}
  \item  If  $\texttt{x}\in\FV({\tt M})$ then, $\mathcal{D}$ includes a subderivation 
    $\mathcal{D}^\texttt{N}$ concluding  $B^\texttt{N}\vdash \texttt{N} :\tau$, for some $B\subseteq B^\texttt{N}$ and $\tau$; and, moreover,
     $B \vdash \lambda \texttt{x}.{\tt M} : \Pi \texttt{x}^\tau.\sigma$.
  \item
    If  $\texttt{x}\not\in\FV({\tt M})$ and $\texttt{x}\not\in\dom(B)$  then $B \vdash \lambda \texttt{x}.{\tt M} :\Pi \texttt{x}^\tau.  \sigma$ for any $\tau$.
\item
    If  $\mathcal{D}^\texttt{N}$ is a derivation concluding  $B\vdash \texttt{N} :\tau$, for some $\tau$ then
    $B \vdash (\lambda \texttt{x}.{\tt M}) \texttt{N}   :\sigma$.
\end{enumerate}
 \end{lemma}
 \begin{proof}
   \begin{enumerate}
   \item
     First, by induction on $\mathcal{D}$ we prove that in $\mathcal{D}$ there is a subderivation $B^\texttt{N}\vdash \texttt{N} :\tau$. 
     Second, by induction on  $\mathcal{D}$ we prove that we can transform $\mathcal{D}$ in a derivation $\mathcal{D}^*$ concluding  $B\cup\{ {\tt x}:\tau  \} \vdash {\tt M} :\sigma$.
     We conclude by using the rule $(P_1)$.
   \item Since ${\tt M}[\texttt{N}/\texttt{x}]={\tt M}$, it is easy to prove that   $B\cup\{ {\tt x}:\tau  \} \vdash {\tt M} :\sigma$ by induction on  $\mathcal{D}$.
     We conclude by using the rule $(P_1)$.
     \item By the previous cases of this Lemma and by using the rule $(P_2)$ (note that ${\tt x}$ does not occur in $\sigma$ by hypothesis).\qed
   \end{enumerate}
 \end{proof}

Let $C[.]$  denote a context for $\qpcf$.
 
 \begin{lemma}[Typed subject expansion]\label{lemma:subjEXP} 
   Let  $\mathcal{D}^\texttt{N}$ be a derivation concluding  $B\vdash \texttt{N} :\tau$, for some $\tau$.
   If $B \vdash C[ {\tt M}[\texttt{N}/\texttt{x}]] :\sigma$ then $B \vdash C[ (\lambda \texttt{x}.{\tt M}) \texttt{N}] :\sigma$.
 \end{lemma}
 \begin{proof}
   W.l.o.g. we can assume that $\texttt{x}$ is fresh, so the proof follows 
   by induction on $B \vdash C[ {\tt M}[\texttt{N}/\texttt{x}]] :\sigma$ and by using the Lemma \ref{lem:auxExp}.
   \end{proof}

It is worth to remark some peculiarity of this typing system.

\begin{lemma}\label{idxProperties}
  \begin{enumerate}
  \item If $B\vdash {\tt M}:\Idx$ then $B\vdash {\tt M}:\Nat$.
\item Let $B\vdash {\tt M}:\sigma$ be a typing derivation.\vphantom{$\surfCirc$}
If $\CIRC{\texttt{E}}$ occurs in it, then $B\vdash\texttt{E}:\Idx$. 
\item If  $B\vdash {\tt M}:\sigma$ then $\wfi{B,\codom(B)}$ and $\wfi{B,\sigma}$.
  \end{enumerate}
\end{lemma}
\begin{proof}
\emph{1.} By rule $(I_0)$.
\emph{2,3.} The proof can proceed by induction on the derivation $B\vdash {\tt M}: \sigma$ by using the Substitution Lemma.
 \qed
\end{proof}

\section{Semantics of $\qpcf$} \label{sec:eval}

As for $\PCF$,  the evaluation of $\qpcf$ focuses on programs, viz. closed terms of ground types.
Nevertheless, $\PCF$ has just one (or two) ground types,
while $\qpcf$ has denumerable closed ground types, namely $\Nat$, $\Idx$ and $\CIRC{\num{n}}$.

As for $\PCF$, the evaluation of a term  typed $\Nat$ can either diverge or give back a numeral.
The evaluation of a term  typed  $\Idx$ should always converge, and  give back a numeral.
The evaluation of a term typed $\CIRC{\texttt{E}}$ can  either diverge or give back an evaluated circuit.
In the following, we use $\mathtt{C}$ to denote circuits  (resulting from the evaluation of circuit expressions),
i.e. strings built by gate-names, parallel composition and serial composition.
Moreover, numerals and circuit strings are sometimes denoted $\mathtt{V}$.

\begin{definition}\label{evaluationRules}
We formalize the evaluation of $\qpcf$ by means of statements of the shape 
$\mathtt{M} \ev^{\PRB{\alpha}}  \mathtt{V}$
obtained as conclusion of a (finite) derivation $\mathcal{D}$  built with the rules in Table~\ref{evaluationRules},
such that: (i) $\mathtt{M}$ is a closed ground term; and, (ii) $0< \alpha\leq 1$ is the probability 
that $\mathcal{D}$ is the evaluation.
For the sake of simplicity, we write $\mathtt{M} \ev  \mathtt{V}$ to mean that there is a derivation (not necessarily unique)
concluding $\mathtt{M} \ev^{\PRB{\alpha}}  \mathtt{V}$ for some $0< \alpha\leq 1$
and, we write $\mathtt{M} \Uparrow$ to mean that no $\PRB{\alpha}> 0$ and $\mathtt{V}$ exist, such that 
$\mathtt{M} \ev^{\PRB{\alpha}}  \mathtt{V}$. 
\end{definition}

We remark that the semantics of Table~\ref{evaluationRules} rests on an external semantics via the premise of the rule $(m)$,
that executes the circuit in accord with the laws of quantum mechanics  (cf. Definition \ref{unitSemantic}).

\begin{table}[t]
$$\begin{myArray}[2.1]{c}
 \infer[(n)]{\num{n}\ev^{\PRB{1}}  \tt \num{n}}{} \qquad \infer[(s)]{\tt \Succ(M)\ev^{\PRB{\alpha}}  \num{n+1}}{\tt M\ev^{\PRB{\alpha}}  \n} \qquad
\infer[(p)]{\tt \Pred(M)\ev^{\PRB{\alpha}}  \n}{\tt M\ev^{\PRB{\alpha}}  \num{n+1}}\qquad 
\infer[(\beta)]{{\tt (\lambda x.M)N}{\tt P}_1\cdots {\tt P}_m\ev^{\PRB{\alpha}}  {\tt V}}{{\tt M[N/x]}{\tt P}_1\cdots {\tt P}_m\ev^{\PRB{\alpha}}  \texttt{V}}\\
\infer[(\mathrm{if}_l)]{\tt \lif{M}{L}{R}\ev^{\PRB{\alpha\cdot\alpha'}}  \texttt{V} }{\tt M\ev^{\PRB{\alpha}}  \num{0} & \tt L\ev^{\PRB{\alpha'}}  \texttt{V}}\qquad
\infer[(\mathrm{if}_r)]{\tt \lif{M}{L}{R}\ev^{\PRB{\alpha\cdot\alpha'}}  \texttt{V}}{\tt M\ev^{\PRB{\alpha}}  \num{n+1} & \tt R\ev^{\PRB{\alpha'}}  \texttt{V}}\qquad
\infer[(Y)]{ \TTY {\tt M} {\tt P}_1\cdots {\tt P}_i\ev^{\PRB{\alpha}}  {\tt V}}{{\tt M (\TTY {\tt M})}{\tt P}_1\cdots {\tt P}_i\ev^{\PRB{\alpha}}  \texttt{V}}\\
\infer[(\mathrm{sz})]{ \size \texttt{M} \ev^{\PRB{\alpha}}  \num{n} }{\vdash \texttt{M}:\CIRC{\texttt{E}} \qquad \texttt{E}\ev^{\PRB{\alpha}} \num{n} } 
       \qquad \infer[(\mathrm{op})]{\odot\,\texttt{E}_0\,\texttt{E}_1 \ev^{\PRB{\alpha\cdot\alpha'}}  \num{m\odot n} }{\texttt{E}_0\ev^{\PRB{\alpha}} \num{m} \quad \texttt{E}_1\ev^{\PRB{\alpha'}} \num{n}}
       \qquad \infer[(\mathrm{gt})]{\getBit\,\texttt{M}\,\texttt{N} \ev^{\PRB{\alpha\cdot\alpha'}}  \bitExtraction{\num{m}}{\num{n}}}{\texttt{M}\ev^{\PRB{\alpha}} \num{m} \qquad \texttt{N}\ev^{\PRB{\alpha'}} \num{n}}\\
\infer[(\mathrm{st})]{\setBit\,\texttt{M}\,\texttt{N} \ev^{\PRB{\alpha\cdot\alpha'}}  \num{m}'}
                        {\texttt{M}\ev^{\PRB{\alpha}} \num{m} \quad \texttt{N}\ev^{\PRB{\alpha'}} \num{n} \quad 
                           \text{ and } \num{m}' \text{ is such that } \bitExtraction{\num{m}'}{\num{n}}=1 \text{ and }\forall \num{k}\neq\num{n} \bitExtraction{\num{m}'}{\num{k}}= \bitExtraction{\num{m}}{\num{k}} } \\
\infer[(\mathrm{u})]{   {\tt U} \ev^{\PRB{\alpha}}   {\tt U}}{ }  \qquad 
\infer[(\mathrm{u'})]{  {\tt M}_0 \typecolon{\tt M}_1  \ev^{\PRB{\alpha\cdot\alpha'}}   {\tt C}_0  \typecolon {\tt C}_1 }
                     { {\tt M}_0\ev^{\PRB{\alpha}}  {\tt C}_0 \quad  {\tt M}_1\ev^{\PRB{\alpha'}}   {\tt C}_1} \qquad
\infer[(\mathrm{u''})]{{\tt M}_0\parallel {\tt M}_1\, \ev^{\PRB{\alpha\cdot\alpha'}}   {\tt C}_1 \parallel {\tt C}_0 }
                     { {\tt M}_0\ev^{\PRB{\alpha}}  {\tt C}_0 \qquad  {\tt M}_1\ev^{\PRB{\alpha'}}  {\tt C}_1 }\\
\infer[(\mathrm{r_0})]{\revCirc {\tt M} \ev^{\PRB{\alpha}}  \mathtt{U}'}{ {\tt M}\ev^{\PRB{\alpha}}  {\tt U} & (\ddagger{\tt U})=\mathtt{U}' } \qquad
\infer[(\mathrm{r_1})]{\revCirc {\tt M} \ev^{\PRB{\alpha\cdot\alpha'\cdot\alpha''}}  {\tt C}'_1  \typecolon {\tt C}'_0 }
            { {\tt M}\ev^{\PRB{\alpha}}  {\tt C}_0  \typecolon {\tt C}_1 & \revCirc{\tt C}_0\ev^{\PRB{\alpha'}}  {\tt C}'_0 & \revCirc{\tt C}_1\ev^{\PRB{\alpha''}}  {\tt C}'_1 }\\
\infer[\!\!(\mathrm{it})]{\iterCirc\, \texttt{E} \,{\tt M}_0\, {\tt M}_1 \ev^{\PRB{\alpha\cdot\alpha'\cdot\alpha''}} 
                                                    \!\!\!\underbrace{{\tt C}_1\!\parallel\!\! \cdots\!\! \parallel\! {\tt C}_1}_{\num{n}}\! \parallel\! {\tt C}_0 }
                            { \texttt{E} \ev^{\PRB{\alpha}} \num{n}\qquad  {\tt M}_0\ev^{\PRB{\alpha'}}  {\tt C}_0 \qquad  {\tt M}_1\ev^{\PRB{\alpha''}}  {\tt C}_1 }\quad 
\infer[\!\!(\mathrm{r_2})]{\revCirc {\tt M} \ev^{\PRB{\alpha\cdot\alpha'\cdot\alpha''}}  {\tt C}'_0  \parallel {\tt C}'_1 }
                     { {\tt M}\ev^{\PRB{\alpha}}\!  {\tt C}_0 \!\parallel\! {\tt C}_1 & \revCirc{\tt C}_0\ev^{\PRB{\alpha'}}\!  {\tt C}'_0 & \revCirc{\tt C}_1\ev^{\PRB{\alpha''}}\!  {\tt C}'_1 }\\
\phantom{\hspace{3.5cm}}
  \infer[(\mathrm{m})]{\dMeas({\tt M},{\tt N}) \ev^{\PRB{\alpha\cdot\alpha'\cdot\alpha''}}  \num{n}}
        { \mathtt{M}\ev^{\PRB{\alpha}} \num{m} \quad \mathtt{N}\ev^{\PRB{\alpha'}}  \mathtt{C} \quad \mathtt{N}:\CIRC{\num{k}}   \qquad  (n,\alpha'')\in\text{circuitEval}^k(\num{m}\restr{\num{k}},{\tt C})  }\\
\hline
\end{myArray}$$
\caption{Operational Semantics.}
\label{evaluationRules}
\end{table}

Table  \ref{evaluationRules}  includes the standard call-by-name semantics of $\PCF$, namely the first two lines of rules are well-known.

The rules $(sz)$, $(op)$ evaluate some index expressions. In particular, $(sz)$ uses the typing information of $\mathtt{M}$  to recover its arity information. 
Since types are preserved during the evaluation (cf. Section \ref{sec:properties}), we can be sure that the information we extract from types is an index.
Moreover, it is strongly normalizing by Theorem \ref{idxNormalization}. 

\begin{example}
  We consider an example that, in some sense, allows us to complement the Example \ref{exampleDependent}.
  It is easy to see that $\vdash\TTY (\lambda \texttt{x}^{\CIRC{8}}.\texttt{x}):\CIRC{8}$ and $\vdash \size(\TTY (\lambda \texttt{x}^{\CIRC{8}}.\texttt{x})): \Idx$.
Although $\TTY (\lambda \texttt{x}^{\CIRC{8}}.\texttt{x}) \Uparrow$, it is interesting to note that  $\size(\TTY (\lambda \texttt{x}^{\CIRC{8}}.\texttt{x}))\ev^1 \num{8}$.\qed
\end{example}

Let $\bitExtraction{\num{m}}{\num{n}}$ be notation for  $(\overbrace{(\num{m}\mathbin{/}\num{2})\ldots\mathbin{/}\num{2} }^{\num{n}})\%\num{2}$ where $\mathbin{/}$ 
is the integer division (neglecting the remainder of the division)
and $\%$ is the modulo (giving back the remainder of the division). Thus, $\bitExtraction{\num{m}}{\num{0}}$ is the rightmost bit of the binary representation of $\num{m}$.
The rules (gt) and (st) get/set a bit of the first argument (the one selected by the second argument).
For example, the numeral $\setBit\,\num{0}\,\num{n+1}$ is the decimal representation of the binary state $1\underbrace{0\ldots 0}_{n}$ and $\getBit\,\num{3}\,\num{0}$ yields the bit $1$.
Clearly, $\setBit,\getBit$ are  syntactic sugar to manage  input states. 

The rules  $\rm (u), (u'), (u''), (r_0),(r_1), (r_2), (it)$ evaluate circuit expressions in  circuits, viz. strings on $\typecolon,\parallel $ and the gate-names $\mathtt{U}$.
Note that $\rm(u')$ and $\rm(u'')$ evaluate sequential and parallel composition of circuit expressions.
\begin{example}[Sequential composition of quantum circuits]\label{ex:seqcirc}
  Let $\texttt{C}:\CIRC{\num{k}}$ be a given circuit.
   We can use $\lambda \mathtt{x}^{\CIRC{\num{k}}}.\mathtt{x}\appCirc\mathtt{x}: \CIRC{\num{k}}\fleche\CIRC{\num{k}}$ 
in order to concatenate two copies of $\texttt{C}$. 
Let $\num{k}$ be an arbitrary numeral and let $\mathsf{M_{seq}}$ be
$$\lambda \mathtt{u}^{\CIRC{\num{k}}}.\lambda \texttt{x}^{\Nat}. Y\mathsf{W}\mathtt{u}\texttt{x}:\CIRC{\num{k}}\fleche \Nat\fleche\CIRC{\num{k}}$$
where $\mathsf{W}=\lambda \mathtt{w}^\sigma.\lambda \mathtt{u}^{\CIRC{\num{k}}}.\lambda \texttt{y}^{\Nat}.
\lif{\texttt{y}}{(u)}{\left(\appCirc \, (u)\, (w\,u\,(\Pred\, \texttt{y})) \right)}$ 
has type $\sigma\fleche\sigma$, with $\sigma=\CIRC{\num{k}}\fleche\Nat\fleche\CIRC{\num{k}}$.
We can use $\mathsf{M_{seq}}$   applied  to ${\tt C}$ and $\num{n}$ to concatenate $n+1$ copies of ${\tt C}$.   
It is straightforward to parameterize $\mathsf{M_{seq}}$  in order to transform it in a template for a circuit-builder that can be used for any arity.
It suffices to replace $\num{k}$ with the  variable ${\tt k}^\Idx$ and to abstract it;
so that the resulting term
$\mathsf{M^{A}_{seq}} = \lambda {\tt k}^\Idx. \lambda \mathtt{u}^{\CIRC{\num{k}}}.\lambda \texttt{x}^{\Nat}. Y\mathsf{W}\mathtt{u}\texttt{x}$
has type $\Pi{\tt k}^\Idx.\CIRC{\mathtt{k}}\fleche \Nat\fleche\CIRC{\mathtt{k}} $. \qed
 \end{example}

 The rule  $\rm(it)$  provides a mechanism to compose  circuits in parallel. It is driven by an argument of type $\Idx$
in order to ensure that iteration is strong normalizing  and, consequently,  that the arity of the generated circuit is always a numeral.

\begin{example}[Parallel composition of quantum circuit]\label{ex:parcirc}
  Let $\mathsf{M_{par}}$ be
  $$\lambda \mathtt{x}^{\Idx}.\lambda \mathtt{u}^{\CIRC{\num{k}}}\lambda \mathtt{w}^{\CIRC{\num{h}}}.
\iterCirc\, \mathtt{x}\, \mathtt{u}\, \mathtt{w}: \Pi{\tt x}^\Idx.\CIRC{\num{k}}\fleche\CIRC{\num{h}}\fleche\CIRC{\num{k}+(\mathtt{x}*(\num{h}+1))}.$$
$\mathsf{M_{par}}$ when applied to an $\num{n}$ and two unitary gates $\mathtt{U}_1:\CIRC{\num{k}}$ and $\mathtt{U}_2:\CIRC{\num{h}}$ 
 generates a simple circuit built upon a copy of gate $\mathtt{U}_1$ in parallel with {$n$} copies of gate $\mathtt{U}_2$.
It is straightforward to parameterize the above example. It suffices to replace numerals $\num{k}$ and $\num{h}$ in the above example by variables and to abstract
to obtain a single parametric term typed
$\Pi{\tt k}^\Idx.\Pi{\tt h}^\Idx.\Pi{\tt x}^\Idx.\CIRC{\mathtt{k}}\fleche\CIRC{\mathtt{h}}\fleche \CIRC{\mathtt{k}+(\mathtt{x}*(\mathtt{h}+1))}$.\qed
 \end{example}

 More recent quantum programming languages \cite{green2013acm,lmtamc,qwire,Ross2015}
 include the possibility to manipulate quantum circuits and, in particular, of reversing circuits.
 Likewise, our operator $\revCirc $ is expected to produce the adjoint circuit of its input. 
 Its definition rests on the choice of total endo-function (mapping each gate of arity $k$ in a gate of arity $k$) that we denote with the symbol $\ddagger$,
 and that returns the adjoint of each gate.
The circuit reversibility is implemented by rewiring gates in reverse order and,
then, by replacing each gate by its adjoint. Rules (r0), (r1) and (r2) implement this policy.

The evaluation of circuits is characterized by Lemma~\ref{lemma:circExp}.


\paragraph{Quantum Co-processor.}  
The rule (m) evaluates the $\dMeas$ arguments and it uses the results of these evaluations to feed an external evaluating device: a quantum co-processor~\cite{quipperacm}. 
It is considered as a black-box that receiving a suitable evaluated circuit for the evaluation and its initialization,
gives back a total measurement executed on the final state.
In contrast with the other quantum programming languages based on the QRAM model, the co-processor of $\qpcf$ is not expected to record states between calls to it.
Therefore, the decoherence issues are limited to its internal operations (cf. Section \ref{sec:qram}).

The circuit evaluation is described in the standard way by means of the Hilbert's spaces  and  von Neumann Measurements \cite{Ish95,NieCh00,KLM07}. 

 \begin{remark}
   Following the standard axiomatization of quantum mechanics, usually proposed in terms of some \emph{postulates} (see~\cite{Zorzi16} for a simple formulation), we include in the language an explicit operator that represents the so called \emph{von Neumann Total Measurement}, a special kind of \emph{projective} measurements (see e.g. \cite{Ish95,NieCh00,KLM07}).  Informally, given a quantum state, total projective measure destroys superposition and returns a classical state, i.e. a sequence of classical bits $\ket{b_1,b_2,\ldots,b_n}$.
   The restriction to total measurement is not too restrictive because of the Principle of Deferred Measurement~\cite[p.186]{NieCh00}:
   \begin{quote}\small
    Measurements can always be moved from an intermediate stage of a quantum circuit to the end of the circuit; if the
measurement results are used at any stage of the circuit then the classically
controlled operations can be replaced by conditional quantum operations.
   \end{quote}
\end{remark}

In accord with Section~\ref{subSect:overview}, we recall that $\mathtt{C}:\CIRC{n}$ aims to represent a quantum circuit operating on $n+1$ qubits.

\begin{definition}\label{unitSemantic} 
  Let $\text{Circ}^n$ be the set of (evaluated) circuits typed $\CIRC{\num{n}}$ and let $N=2^{n+1}$.
 Let $\mathcal{H}^N$ be a  Hilbert space of finite dimension $N$, let  $\{\ket{\varphi_i}\}$ be a  orthonormal basis  on $\mathcal{H}^N$ and
 let $\mathcal{H}^N \rightarrow\mathcal{H}^N$ be the set of unitary operators on $\mathcal{H}^N$.
  \begin{enumerate}
  \item   $\HEval^n: \text{Circ}^n \rightarrow (\mathcal{H}^N \rightarrow\mathcal{H}^N)$ is a mapping from evaluated circuits into their corresponding algebraic operators
defined  as follows:
$\HEval^n(\mathtt{U})\myDef\mathbf{U}$ whenever $\mathtt{U}$  is  typed $\CIRC{\num{n}}$, so that  $\mathbf{U}: \mathcal{H}^N \rightarrow\mathcal{H}^N$;
$\HEval^n({\tt C}_0 \parallel {\tt C}_1)\myDef\HEval^{n_0}({\tt C}_0)\otimes \HEval^{n_1}({\tt C}_1)$ whenever ${\tt C}_i$ is typed $\CIRC{\num{n}_i}$ and $n=n_0+n_1$;
$\HEval^n({\tt C}_0 \typecolon {\tt C}_1)\myDef\HEval^n({\tt C}_0)\circ \HEval^n({\tt C}_1)$. 
\item A von Neumann measurement (see page 49 of \cite{KLM07}, for instance) with respect to the basis of $\mathcal{H}^N$  and a given state
$$\psi = \sum_i \alpha_i\ket{\varphi_i}$$
outputs the $i$ with probability $|\alpha_i|^2\in (0,1]$.
\item $\text{circuitEval}^n: \Nat \times \text{Circ}^n\rightarrow 2^{\mathbb{N} \times (0,1]}$ is a mapping from a pair  
(an initial state and a circuit) to a powerset of pairs 
(a vector of the basis and its probability) defined as follows:
$$\text{circuitEval}^n(x,\mathtt{C})=\left\{ (i, |\alpha_i|^2)  \left|  \HEval^n(\mathtt{C})(\ket{\varphi_x}) =  \sum_i \alpha_i\ket{\varphi_i} \right\}\right.\;.$$
 \end{enumerate}
\end{definition}

The rule $(m)$ describes a call to an external quantum co-processor that has to be able to evaluate quantum circuits. 
The co-processor is not assumed to store any quantum state between calls.

\begin{lemma}\label{lemma:circExp}
  \begin{enumerate}
  \item If $\vdash \mathtt{M}:\CIRC{\num{n}}$ and $\mathtt{M} \ev  \mathtt{N}$ then   $\vdash \mathtt{N}:\CIRC{\num{n}}$ where $\mathtt{N}$ is a circuit.
  \item If $\mathtt{C}$ is a circuit such that $\vdash \mathtt{C}:\CIRC{\num{n}}$ then $\HEval^n(\mathtt{C})$ is well-defined.
\end{enumerate}
\end{lemma}
\begin{proof}
  We recall that circuits are built on $\mathtt{U} \mid  \appCirc \mid\;\parallel$.
  \begin{enumerate}
  \item  $\mathtt{M} \ev  \mathtt{N}$ means that there exists a derivation $\mathcal{D}$ concluding $\mathtt{M} \ev^{\PRB{\alpha}}  \mathtt{V}$ where $\alpha >1$.
    The proof is  by induction on the last rule applied in $\mathcal{D}$. 
    Rules $(n)$, $(s)$, $(p)$,$(sz)$, $(op)$, $(gt)$, $(st)$ and $(m)$ are not possible because of the typing hypothesis.
    Rules $(\beta)$, $(\mathrm{if}_l)$, $(\mathrm{if}_r)$ and $(Y)$ follow immediately by the induction hypothesis.
    The rule $(u)$ is immediate. The rule $(r_0)$ is straightforward,  because   $\ddagger$ (cf. rule $(r_0)$) is a total operator that respects arities.
    The rules $(u')$, $(u'')$, $(r_1)$, $(r_2)$, $(it)$ follow immediately by induction. 
  \item The proof follows by induction on the typing rules. 
By the previous point of this lemma, it is sufficient to check the rules $(C_1),(C_2),(C_3)$. \qed 
\end{enumerate}
\end{proof}


Let $\num{m}\restr{\num{k}}$ be the restriction  of the binary representation of $\num{m}$ to the first $\num{k}$ bits.

\begin{example}[Creating and evaluating the EPR Circuit]\label{example:EPR}
Let $\mathtt{CNOT}:\CIRC{\num{1}}$ be the cnot gate and
let $\mathtt{H}:\CIRC{\num{0}}$ and $\mathtt{I}:\CIRC{\num{0}}$  be the (unary) Hadamard and Identity gates, respectively.
The following $\qpcf$-term represents a simple circuit that generates the well-known EPR state~\cite{NieCh00}:
$$\mathtt{EPR}=(\mathtt{I}\parallel\mathtt{H})\appCirc{\text{CNOT}}: \CIRC{\num{1}}\quad .$$
Given a sequence $(b_1,\ldots,b_k)$ of bits, we write $\mathbf{n}(b_1,\ldots,b_k)$ to denote the corresponding numeral.\\
The evaluation of $\dMeas(\num{0},\mathtt{EPR})$ asks to the quantum co-processor the execution of the circuit
$\mathtt{EPR}$ on the initial state $\mathbf{n}(00)$. This execution returns  a fair superposition of $\ket{00}$ and $\ket{11}$ ,
i.e. the state $\frac{1}{\sqrt{2}}\ket{00}+\frac{1}{\sqrt{2}}\ket{11}$. Since
$$\text{circuitEval}(0,\mathtt{EPR})=\{(\mathbf{n}(00), \frac{1}{2}),(\mathbf{n}(11), \frac{1}{2}) \}\quad,$$ 
it follows that  both $\dMeas(\num{0},\mathtt{{EPR}})\ev^{\PRB{\frac{1}{2}}}\num{0}$ 
and  $\dMeas(\num{0},\mathtt{{EPR}})\ev^{\PRB{\frac{1}{2}}} \num{3}$ 
(note that $\num{0}=\mathbf{n}(0,0)$ and $\num{3}=\mathbf{n}(1,1)$).
 \qed
\end{example}

It is well-known that quantum measures break the deterministic evolution of a quantum system. As a consequence,
in presence of a measurement operator in a quantum language (equipped with an universal basis of quantum gates),
one  necessarily loses confluence. For details on this argument see~\cite{DLMZ11entcs,diaz11}.

\subsection{On the Probabilistic evaluation}

The probabilistic evaluation of formal quantum programming languages is usually defined by means of a small-step operational semantics
(e.g. \cite{selinger2006mscs,selinger09chap,pagani14acm}) that formalizes the desired reduction strategy.
Each reduction rule is labeled with the probability that the reduction fires. This essentially means
that the reduction strategy is associated to a discrete time Markov chain, whose states are terms and stationary states are evaluated terms. 
Moreover, various probabilistic and non-deterministic extensions of $\PCF$ have been proposed in literature,
see \cite{ehrhard2018jacm,escardo2009mfps,danos2011iandc,goubault2011ieee,goubault2015jlam}.
The more recent proposals, i.e. \cite{ehrhard2018jacm,danos2011iandc,goubault2015jlam} formalize the language evaluation in accord with the above quantum approach.
In order  to stress the strict correspondence of $\qpcf$ with $\PCF$, we define the evaluation in terms of a big-step operational semantics
that hides the single-step details  (unessential for our purposes) that follows this common approach. 
(Big-step semantics are considered in \cite{escardo2009mfps,goubault2011ieee,rios2017eptcs}.)


In all the above operational evaluation there might be many evaluations from  $\mathtt{M}$ to $\mathtt{V}$.
Let $\mathbf{D}(\mathtt{M},\mathtt{V})$ denote the set of derivations proving $\mathtt{M} \ev^{\PRB{\alpha}} \mathtt{V}$ for any $\alpha$;
let $\mathbf{D}(\mathtt{M})$ denote the set of derivations proving $\mathtt{M} \ev^{\PRB{\alpha}} \mathtt{V}$ for any $\alpha$ and any $\mathtt{V}$;
and, let $\mathsf{prb}(\mathcal{D})$  denotes the probability $\PRB{\alpha}$ whenever $\mathcal{D}$ concludes  $\mathtt{M} \ev^{\PRB{\alpha}}  \mathtt{V}$.
Then,
$\displaystyle\sum_{\mathcal{D}_i\in \mathbf{D}(\mathtt{M},\mathtt{V})} \hspace{-3mm} \mathsf{prb}(\mathcal{D}_i)\leq 1$
and $ \displaystyle\sum_{\mathcal{D}_i\in \mathbf{D}(\mathtt{M})} \hspace{-3mm}\mathsf{prb}(\mathcal{D}_i) \leq 1$ is the probability that the evaluation of $\mathtt{M}$ stops,
while  $1-\displaystyle\sum_{\mathcal{D}_i\in \mathbf{D}(\mathtt{M})} \hspace{-3mm}\mathsf{prb}(\mathcal{D}_i)$ is
the probability that the evaluation of $\mathtt{M}$ diverges.
\begin{example}
  \begin{itemize}
  \item 
    Let $\Omega^\Nat$ be $\TTY (\lambda \mathtt{x}^{\Nat}.\mathtt{x})$.\\ It is clear that $\Omega^\Nat$ is a closed term typed $\Nat$ such that  $\Omega^\Nat\Uparrow$.
    Therefore, $\mathbf{D}(\Omega^\Nat, \num{n})= \emptyset$ and $\mathbf{D}(\Omega^\Nat)= \emptyset$.
\item  Let $\mathtt{EPR}$ be the term defined in the Example \ref{example:EPR}.\\
  Let $\mathtt{M}_8$ be $\ifz \, (\dMeas(\num{0},\mathtt{{EPR}}))\, \num{8} \, (\TTY (\lambda \mathtt{x}^\Nat.\mathtt{x}))$
  then $\mathbf{D}(\mathtt{M}_8)$ contains only the derivation concluding $\mathtt{M}_8\ev^{\PRB{\frac{1}{2}}} \num{8}$.
 \item Let $\mathtt{M}_8^\infty$ be $ \TTY (\lambda \mathtt{x}^\Nat. \ifz \, (\dMeas(\num{0},\mathtt{{EPR}}))\; \num{8} \; \mathtt{x})$
  then $\mathbf{D}(\mathtt{M}_8^\infty)$ contains denumerable derivations: for each $k\in\mathbb{N}$ there is a derivation concluding $\mathtt{M}^\infty_8\ev^{(\PRB{\frac{1}{2}})^k} \num{8}$.
\qed\end{itemize}
\end{example}
        
We define the observational equivalence of $\qpcf$ by following  \cite{ehrhard2018jacm}.

\begin{definition}
 Let $\mathtt{M}$ and $\mathtt{M}$ be closed terms of the same type.
 They are observationally equivalent whenever, for each context $C[.]$ it holds that:
 (i) if $B\vdash C[\mathtt{M}]:\Nat$ then $B\vdash C[\mathtt{N}]:\Nat$; (ii) if $B\vdash  C[\mathtt{N}]:\Nat$ then $B\vdash  C[\mathtt{M}]:\Nat$; 
and, (iii) $\displaystyle\sum_{\mathcal{D}_i\in \mathbf{D}(\mathtt{M},\num{0})} \hspace{-3mm} \mathsf{prb}(\mathcal{D}_i) = \sum_{\mathcal{D}_i\in \mathbf{D}(\mathtt{N},\num{0})} \hspace{-3mm} \mathsf{prb}(\mathcal{D}_i)$.
\end{definition}

As noted in  \cite{ehrhard2018jacm}, the last constraint can be replaced by, for each numeral $\num{k}$,
$\displaystyle\sum_{\mathcal{D}_i\in \mathbf{D}(\mathtt{M},\num{k})} \hspace{-3mm} \mathsf{prb}(\mathcal{D}_i) = \sum_{\mathcal{D}_i\in \mathbf{D}(\mathtt{N},\num{k})} \hspace{-3mm} \mathsf{prb}(\mathcal{D}_i)$.
The operational equivalence is defined  on closed terms of type $\Nat$ because the operational differences in the other types can be 
 traced back to $\Nat$ (while the reverse can be easily proved be false).

Anyway, in this paper we do not plan to further study the observational equivalence of $\qpcf$, therefore in the following,
we sometimes use $\mathtt{M} \ev  \mathtt{V}$ (i.e. $\mathtt{M} \ev^{\PRB{\alpha}}  \mathtt{V}$ for some $0< \alpha\leq 1$).

\subsection{Evaluation properties}\label{sec:properties}

To prove properties about the evaluation mechanisms of $\qpcf$, we have to consider infinite 
ground types (viz. $\Nat$, $\Idx$, $\CIRC{\num{n}}$) 
and we have to unravel the mutual relationships that hold between syntactic classes
(cf. Example~\ref{exampleDependent}). 
A first goal is to prove that $\Idx$ picks up a class of terms which is expected to be endowed with an always terminating evaluation, 
i.e. always normalizing. 

Simply minded arguments do not work for showing the strong normalization property of the typed lambda-calculi:
reduction increases the size of terms, which precludes an induction on their size, and preserves their types, which seems to preclude an induction on types.
We follow the well-known Tait's idea for proofs of strong normalizations  based on a suitable predicate,
see \cite{luca06} for instance. More precisely, we prove our property by adapting the computability predicate given in \cite{gaboardi2016,plotkin77tcs,paolini2006iandc} for $\PCF$.


\begin{definition}\label{compDefinition}
  The predicate $Comp(B,{\tt M},\sigma)$ holds  whenever $B\vdash {\tt M}:\sigma$ and one of the following cases is satisfied:
\begin{enumerate}
\item $B=\emptyset$, $\sigma=\Nat$;  
\item $B=\emptyset$,  $\sigma=\Idx$ and ${\tt M}\ev^1 \num{n}$, for some $\num{n}$;
  \item $B=\emptyset$, $\sigma=\CIRC{\mathtt{E}}$ and $Comp(\emptyset,{\tt E},\Idx)$;
\item $B=\emptyset$,  $\sigma=\Pi \mathtt{x}^\mu. \tau$ and $Comp(\emptyset,{\tt MN}, \tau[{\tt N}/\mathtt{x}])$, for all  $Comp(\emptyset,{\tt N},\mu)$;
\item $B=\{{\tt x}:\nu\}\cup B'$ implies $Comp(B'[{\tt N}/{\tt x}],{\tt M}[{\tt N}/{\tt x}], \sigma[{\tt N}/{\tt x}])$ for all $Comp(\emptyset,{\tt N},\nu)$.
\end{enumerate}
\end{definition}

Let us assume $B\vdash {\tt M}:\sigma$ holds. Focus on ground types:
(i) $Comp$ always holds for $\Nat$;  (ii) $Comp$ holds for $\Idx$ when the term evaluation is terminating;
(iii)$Comp$ holds for $\CIRC{\mathtt{E}}$ independently from the typed term,  whenever  $\mathtt{E}$ is well typed and its evaluation is terminating.
The remaining cases ensure that  $Comp$ hold: (i) for all well-typed closing substitution of $\tt M$  and, (ii) for all well-typed application.

\begin{notation}
Let ${\tt x}_1,...,{\tt x}_n$ be variables, let ${\tt N}_1,...,{\tt N}_n$ be terms and let  $1\leq j\leq k\leq n$.
We write  $\mathtt{Q}\vv{[\mathtt{N}/\mathtt{x}]_j^{k}} $ as a shortening for $((\mathtt{Q}[{\tt N}_j/{\tt x}_j])\cdots [{\tt N}_k/{\tt x}_k] )$
when $\mathtt{Q}$ is a term.
We write  $\sigma\vv{[\mathtt{N}/\mathtt{x}]_j^{k}} $ as a shortening for $((\sigma[{\tt N}_j/{\tt x}_j])\cdots [{\tt N}_k/{\tt x}_k] )$
when $\sigma$ is a type.
As expected $k< j$ means no substitution.
\end{notation}

\begin{lemma}\label{compRephrased}$\;$
\begin{enumerate}
\item Let $\kappa=\Pi\mathtt{z}_1^{\tau_1}. ... \Pi\mathtt{z}_m^{\tau_m}. \Nat$;
  $Comp(B,{\tt M},\kappa)$  iff $B=\{{\tt x}_1:\nu_1,...,{\tt x}_n:\nu_n\}$ implies $\vdash{\tt M}\vv{[\mathtt{N}/\mathtt{x}]_1^{n}} \mathtt{P}_1 \ldots \mathtt{P}_m : \Nat$,
  for all ${\tt N}_j$  such that $Comp(\emptyset,{\tt N}_j,\nu_j\vv{[\mathtt{N}/\mathtt{x}]_1^{j-1}})$ where $j\leq n$,
  for all ${\tt P}_i$  such that $Comp(\emptyset,{\tt P}_i,\tau_i\vv{[\mathtt{N}/\mathtt{x}]_1^{n}}\vv{[\mathtt{P}/\mathtt{z}]_1^{i-1}} )$ where $i\leq m$.
\item Let $\kappa=\Pi\mathtt{z}_1^{\tau_1}. ... \Pi\mathtt{z}_m^{\tau_m}. \Idx$; 
  $Comp(B,{\tt M},\kappa)$  iff $B=\{{\tt x}_1:\nu_1,...,{\tt x}_n:\nu_n\}$ implies  $Comp(\emptyset,{\tt M} \vv{[\mathtt{N}/\mathtt{x}]_1^{n}}\mathtt{P}_1 \ldots \mathtt{P}_m , \Idx)$,
   for all ${\tt N}_j$  s.t. $Comp(\emptyset,{\tt N}_j,\nu_j\vv{[\mathtt{N}/\mathtt{x}]_1^{j-1}})$ where $j\leq n$,
  for all ${\tt P}_i$  s.t. $Comp(\emptyset,{\tt P}_i,\tau_i  \vv{[\mathtt{N}/\mathtt{x}]_1^{n}}\vv{[\mathtt{P}/\mathtt{z}]_1^{i-1}} )$ where $i\leq m$.
\item Let $\kappa=\Pi\mathtt{z}_1^{\tau_1}. ... \Pi\mathtt{z}_m^{\tau_m}.\CIRC{\mathtt{E}}$.
  $Comp(B,{\tt M},\kappa)$ iff  $B=\{{\tt x}_1:\nu_1,...,{\tt x}_n:\nu_n\}$ implies
  $\vdash {\tt M}\vv{[\mathtt{N}/\mathtt{x}]_1^{n}} \mathtt{P}_1 \ldots \mathtt{P}_m : \CIRC{\mathtt{E}\vv{[\mathtt{N}/\mathtt{x}]_1^{n}}\vv{[\mathtt{P}/\mathtt{z}]_1^{m}}  }$ and
   $Comp(\emptyset,{\tt E} \vv{[\mathtt{N}/\mathtt{x}]_1^{n}}\vv{[\mathtt{P}/\mathtt{z}]_1^{m}} , \Idx)$,
   for all ${\tt N}_j$  such that $Comp(\emptyset,{\tt N}_j,\nu_j \vv{[\mathtt{N}/\mathtt{x}]_1^{j-1}} )$ where $j\leq n$,
  for all ${\tt P}_i$  such that $Comp(\emptyset,{\tt P}_i,\tau_i \vv{[\mathtt{N}/\mathtt{x}]_1^{n}}\vv{[\mathtt{P}/\mathtt{z}]_1^{i-1}} )$ where $i\leq m$.
\end{enumerate}
\end{lemma}
\begin{proof}
\begin{enumerate}
\item 
  First, we prove by induction on $n$ that 
  $Comp(B,{\tt M},\kappa)$ if and only if  $B=\{{\tt x}_1:\nu_1,...,{\tt x}_n:\nu_n\}$ implies
$Comp(\emptyset,{\tt M} \vv{[\mathtt{N}/\mathtt{x}]_1^{n}},\kappa \vv{[\mathtt{N}/\mathtt{x}]_1^{n}} )$
for all ${\tt N}_j$  such that $Comp(\emptyset,{\tt N}_j,\nu_j \vv{[\mathtt{N}/\mathtt{x}]_1^{j-1}} )$ where $j\leq n$.
Then, we conclude by induction on $m$.
\item The proof is similar to the previous one. It is worth to emphasize that, in the statement, we write
  $Comp(\emptyset,{\tt M}\vv{[\mathtt{N}/\mathtt{x}]_1^{n}} \mathtt{P}_1 \ldots \mathtt{P}_m , \Idx)$  
 as a shortening for  $\vdash{\tt M}\vv{[\mathtt{N}/\mathtt{x}]_1^{n}} \mathtt{P}_1 \ldots \mathtt{P}_m:\Idx$ and
  $\vdash{\tt M}\vv{[\mathtt{N}/\mathtt{x}]_1^{n}}\mathtt{P}_1 \ldots \mathtt{P}_m\ev^1 \num{n}$, for some $\num{n}$.
\item Similar to  that of 1.\qed
\end{enumerate}
\end{proof}
  
Lemma~\ref{compRephrased} provides an alternative charactrerization of $Comp$, because
 it is easy to check that each type has the shape $\Pi\mathtt{z}_1^{\tau_1}. ... \Pi\mathtt{z}_m^{\tau_m}.\gamma$  for a unique $m\in\mathbb{N}$
and $\gamma \in\{\Nat, \Idx, \CIRC{\mathtt{E}}\}$.

\begin{theorem}\label{thComp}
 If $B\vdash \mathtt{M} : \kappa$ then $Comp(B,{\tt M},\kappa)$.
\end{theorem}
\begin{proof}
   The proof is by induction on the derivation $\mathcal{D}$ proving $B\vdash \mathtt{M} : \kappa$.   
\begin{itemize}
\item   Rule $(P_0)$. Let $\mathtt{M}=\mathtt{x}_k$ for some $k\leq n$ and $\kappa=\Pi\mathtt{z}_1^{\tau_1}. ... \Pi\mathtt{z}_m^{\tau_m}. \gamma$ where
  $\gamma \in\{\Nat, \Idx, \CIRC{\mathtt{E}}\}$. 
  If $B=\{{\tt x}_1:\nu_1,...,{\tt x}_n:\nu_n\}$ then we aim to prove that
  $Comp(\emptyset, \mathtt{x}_k\vv{[\mathtt{N}/\mathtt{x}]_1^{n}},\kappa \vv{[\mathtt{N}/\mathtt{x}]_1^{n}} )$
  for all ${\tt N}_j$  such that $Comp(\emptyset,{\tt N}_j,\nu_j \vv{[\mathtt{N}/\mathtt{x}]_1^{j-1}} )$ where $j\leq n$. 
  Since $Comp(\emptyset,{\tt N}_j,\nu_j \vv{[\mathtt{N}/\mathtt{x}]_1^{j-1}} )$ implies $\vdash {\tt N}_j:\nu_j \vv{[\mathtt{N}/\mathtt{x}]_1^{j-1}} $,
  we are sure that ${\tt N}_j$ and $\nu_j \vv{[\mathtt{N}/\mathtt{x}]_1^{j-1}}$ do not contain free variables. 
  Therefore  $Comp(\emptyset, \mathtt{x}_k\vv{[\mathtt{N}/\mathtt{x}]_1^{n}},\kappa \vv{[\mathtt{N}/\mathtt{x}]_1^{n}} )  = Comp(\emptyset, \mathtt{N}_k,\kappa\vv{[\mathtt{N}/\mathtt{x}]_1^{n}} )$.
  Since  $(P_0)$ is the last rule used in $\mathcal{D}$, then its final typing has shape $B''\cup \{  \mathtt{x}_k: \nu_k \} \vdash \mathtt{x}_k: \kappa$
where $ \nu_k = \kappa$.
  Then, 
  $ Comp(\emptyset, \mathtt{N}_k,\kappa \vv{[\mathtt{N}/\mathtt{x}]_1^{n}} ) =Comp(\emptyset, \mathtt{N}_k, \nu_j \vv{[\mathtt{N}/\mathtt{x}]_1^{j-1}})$ which is assumed to hold.
\item Rule $(P_1)$. Let $\mathtt{M}={\tt \lambda x^{\sigma}.Q}$ and $\kappa=\Pi\mathtt{x}^{\sigma}.\Pi\mathtt{z}_1^{\tau_1}. ... \Pi\mathtt{z}_m^{\tau_m}. \gamma$ where
  $\gamma$ is a ground type. 
  By induction, $Comp(B\cup  \{ {\tt x}:\sigma\},\mathtt{Q},\Pi\mathtt{z}_1^{\tau_1}. ... \Pi\mathtt{z}_m^{\tau_m}. \gamma) $ holds. 
 \begin{itemize}
 \item
   Let $\gamma=\Nat$;  by applying the Lemma~\ref{compRephrased} to the induction hypothesis we know that
$B\cup  \{ {\tt x}:\sigma\}=\{{\tt x}_1:\nu_1,...,{\tt x}_n:\nu_n\}$ implies $\vdash{\tt M}\vv{[\mathtt{N}/\mathtt{x}]_1^{n}} \mathtt{P}_1 \ldots \mathtt{P}_m : \Nat$,
  for all ${\tt N}_j$  such that $Comp(\emptyset,{\tt N}_j,\nu_j\vv{[\mathtt{N}/\mathtt{x}]_1^{j-1}})$ where $j\leq n$,
  for all ${\tt P}_i$  such that $Comp(\emptyset,{\tt P}_i,\tau_i\vv{[\mathtt{N}/\mathtt{x}]_1^{n}}\vv{[\mathtt{P}/\mathtt{z}]_1^{i-1}} )$ where $i\leq m$.
  Let $k \leq n$ be such that $ {\tt x}:\sigma$ is  ${\tt x}_k:\nu:_k$.
  So
  $\vdash ({\tt \lambda x^{\sigma}.Q})\vv{[\mathtt{N}/\mathtt{x}]_1^{k-1}} \vv{[\mathtt{N}/\mathtt{x}]_{k+1}^{n}} \mathtt{N}_k \mathtt{P}_1 \ldots \mathtt{P}_m : \Nat$
   follows by Lemma~\ref{lemma:subjEXP}, since $\mathtt{N}_k$ is closed.
   \item 
     Let $\gamma=\Idx$; by applying the Lemma~\ref{compRephrased} to the induction hypothesis we know that  
   if we assume  $B\cup  \{ {\tt x}:\sigma\}=\{{\tt x}_1:\nu_1,...,{\tt x}_n:\nu_n\}$ then
     $Comp(\emptyset,{\tt M} \vv{[\mathtt{N}/\mathtt{x}]_1^{n}}\mathtt{P}_1 \ldots \mathtt{P}_m , \Idx)$,
   for all ${\tt N}_j$  such that $Comp(\emptyset,{\tt N}_j,\nu_j\vv{[\mathtt{N}/\mathtt{x}]_1^{j-1}})$ where $j\leq n$,
  for all ${\tt P}_i$  such that $Comp(\emptyset,{\tt P}_i,\tau_i  \vv{[\mathtt{N}/\mathtt{x}]_1^{n}}\vv{[\mathtt{P}/\mathtt{z}]_1^{i-1}} )$ where $i\leq m$.
  More explicitly,
  $\vdash{\tt M} \vv{[\mathtt{N}/\mathtt{x}]_1^{n}} \mathtt{P}_1 \ldots \mathtt{P}_m:\Idx$
  and   ${\tt M} \vv{[\mathtt{N}/\mathtt{x}]_1^{n}} \mathtt{P}_1 \ldots \mathtt{P}_m\ev^1 \num{n}$, for some $\num{n}$.
Let $k \leq n$ be such that $ {\tt x}:\sigma$ is  ${\tt x}_k:\nu:_k$.
  Since $\mathtt{N}_k$ is closed, 
  we conclude by  Lemma~\ref{lemma:subjEXP} and the evaluation rule $\beta$.
\item Let $\gamma= \CIRC{\mathtt{E}}$;  by applying the Lemma~\ref{compRephrased} to the induction hypothesis we know that  
 $B\cup  \{ {\tt x}:\sigma\}=\{{\tt x}_1:\nu_1,...,{\tt x}_n:\nu_n\}$ implies
 $\vdash {\tt M}\vv{[\mathtt{N}/\mathtt{x}]_1^{n}} \mathtt{P}_1 \ldots \mathtt{P}_m : \CIRC{\mathtt{E}\vv{[\mathtt{N}/\mathtt{x}]_1^{n}}\vv{[\mathtt{P}/\mathtt{z}]_1^{m}}  }$ and
   $Comp(\emptyset,{\tt E} \vv{[\mathtt{N}/\mathtt{x}]_1^{n}}\vv{[\mathtt{P}/\mathtt{z}]_1^{m}} , \Idx)$,
   for all ${\tt N}_j$  such that $Comp(\emptyset,{\tt N}_j,\nu_j \vv{[\mathtt{N}/\mathtt{x}]_1^{j-1}} )$ where $j\leq n$,
   for all ${\tt P}_i$ s.t. $Comp(\emptyset,{\tt P}_i,\tau_i \vv{[\mathtt{N}/\mathtt{x}]_1^{n}}\vv{[\mathtt{P}/\mathtt{z}]_1^{i-1}} )$ where $i\leq m$.
   Let $k \leq n$ be such that $ {\tt x}:\sigma$ is  ${\tt x}_k:\nu:_k$.
   Since $Comp(\emptyset,{\tt E}[\mathtt{P}/\mathtt{x},\vv{\mathtt{N}}/\vv{\mathtt{x}},\mathtt{P}_1 /{\tt z}_1 ,...,\mathtt{P}_m/{\tt z}_m], \Idx)$ already holds,
   it remains to prove that
   $\vdash (\lambda\mathtt{x}.{\tt Q}) )\vv{[\mathtt{N}/\mathtt{x}]_1^{k-1}} \vv{[\mathtt{N}/\mathtt{x}]_{k+1}^{n}} \mathtt{N}_k \mathtt{P}_1 \ldots \mathtt{P}_m : \CIRC{\mathtt{E}\vv{[\mathtt{N}/\mathtt{x}]_1^{n}}\vv{[\mathtt{P}/\mathtt{z}]_1^{m}}  }$. This latter follows by  Lemma~\ref{lemma:subjEXP} because     $\mathtt{N}_k$ is closed.
   Thus we conclude by  Lemma~\ref{compRephrased}.
\end{itemize}
\item Rule $(P_2)$. Let $\mathtt{M}=\mathtt{P}\mathtt{Q}$ and $\kappa=\Pi\mathtt{z}_1^{\tau_1}. ... \Pi\mathtt{z}_m^{\tau_m}. \gamma$ where
  $\gamma$ is a ground type.  
  Let $\sigma$ be the type such that  $Comp(B,\mathtt{P},\Pi\mathtt{y}^{\sigma}.\kappa)$ and  $Comp(B,\mathtt{Q},{\sigma})$ hold by induction.
Thus, if
  $B=\{{\tt x}_1:\nu_1,...,{\tt x}_n:\nu_n\}$ then, we have
  $Comp(\emptyset,{\tt P} \vv{[\mathtt{N}/\mathtt{x}]_1^{n}},(\Pi\mathtt{y}^{\sigma}.\kappa) \vv{[\mathtt{N}/\mathtt{x}]_1^{n}} )$
  for all ${\tt N}_j$  s.t. $Comp(\emptyset,{\tt N}_j,\nu_j \vv{[\mathtt{N}/\mathtt{x}]_1^{j-1}} )$ where $j\leq n$.
  But $Comp(\emptyset,{\tt Q} \vv{[\mathtt{N}/\mathtt{x}]_1^{n}},\sigma \vv{[\mathtt{N}/\mathtt{x}]_1^{n}} )$ holds too.
Thus, by Definition~\ref{compDefinition}, we have
 $Comp(\emptyset ,(\mathtt{P}\vv{[\mathtt{N}/\mathtt{x}]_1^{n}}) \mathtt{Q}\vv{[\mathtt{N}/\mathtt{x}]_1^{n}},(\kappa\vv{[\mathtt{N}/\mathtt{x}]_1^{n}})[\mathtt{Q}\vv{[\mathtt{N}/\mathtt{x}]_1^{n}}/\mathtt{y}])$ that, in it is turn, implies
 $ Comp(\emptyset ,(\mathtt{P}\mathtt{Q})\vv{[\mathtt{N}/\mathtt{x}]_1^{n}}, (\kappa [\mathtt{Q}/\mathtt{y}])\vv{[\mathtt{N}/\mathtt{x}]_1^{n}})$.
\item Rule $(P_3)$. By  Lemma~\ref{compRephrased}, we have to prove that $B=\{{\tt x}_1:\nu_1,...,{\tt x}_n:\nu_n\}$ implies
  $\vdash (\Pred \vv{[\mathtt{N}/\mathtt{x}]_1^{n}}) \mathtt{P} : \Nat$,
  for all ${\tt N}_j$  such that $Comp(\emptyset,{\tt N}_j,\nu_j\vv{[\mathtt{N}/\mathtt{x}]_1^{j-1}})$ where $j\leq n$,
  and  ${\tt P}$  such that $Comp(\emptyset,{\tt P},\Nat )$. Namely, $\vdash \Pred\, \mathtt{P} : \Nat$ has to hold whenever $\vdash {\tt P}:\Nat$ holds.
  This is true by rules $(P_2)$ and $(P_3)$. 
\item Rules $(P_4),(P_5),(B_1),(B_2)$ are similar to the previous case.
\item Rule $(P'_5)$. Let ${\tt M} = \ifz$ and $\kappa=\Pi\mathtt{z}_1^{\Nat}.\Pi\mathtt{z}_2^{\CIRC{E}}.\Pi\mathtt{z}_3^{\CIRC{E}}.\CIRC{E} $.
  By  Lemma~\ref{compRephrased},
  we have to prove that  $B=\{{\tt x}_1:\nu_1,...,{\tt x}_n:\nu_n\}$ implies  
  $\vdash (\!\!\ifz\vv{[\mathtt{N}/\mathtt{x}]_1^{n}})\mathtt{P}_1 \mathtt{P}_2 \mathtt{P}_3: \CIRC{\mathtt{E}\vv{[\mathtt{N}/\mathtt{x}]_1^{n}}\vv{[\mathtt{P}/\mathtt{z}]_1^{3}}  }$
  and  $Comp(\emptyset,{\tt E} \vv{[\mathtt{N}/\mathtt{x}]_1^{n}}  \vv{[\mathtt{P}/\mathtt{z}]_1^{3}}, \Idx)$
  for   $Comp(\emptyset,{\tt P}_1,\Nat)$,\\
         $Comp(\emptyset,{\tt P}_2,\CIRC{E\vv{[\mathtt{N}/\mathtt{x}]_1^{n}}[ \mathtt{P}_1/\mathtt{z}_1] })$, 
         $Comp(\emptyset,{\tt P}_3,\CIRC{E\vv{[\mathtt{N}/\mathtt{x}]_1^{n}} \vv{[\mathtt{P}/\mathtt{z}]_1^{3}} })$ and
         for all ${\tt N}_j$  such that $Comp(\emptyset,{\tt N}_j,\nu_j \vv{[\mathtt{N}/\mathtt{x}]_1^{j-1}} )$ where $j\leq n$.
         The typing follows by rule $(P_2)$ and $(P'_5)$,
          since the comp-hypothesis implies that $\vdash {\tt P}_1:\Nat$, $\vdash{\tt P}_2:\CIRC{E\vv{[\mathtt{N}/\mathtt{x}]_1^{n}}[ \mathtt{P}_1/\mathtt{z}_1] }$
         and $\vdash {\tt P}_3:\CIRC{E\vv{[\mathtt{N}/\mathtt{x}]_1^{n}} [ \mathtt{P}_1/\mathtt{z}_1][ \mathtt{P}_2/\mathtt{z}_2] }$.
         Moreover, by induction on $B\vdash \texttt{E}:\Idx$ we have $Comp(B,  \texttt{E}, \Idx)$,
         therefore $Comp(\emptyset,{\tt E} \vv{[\mathtt{N}/\mathtt{x}]_1^{n}} [ \mathtt{P}_1/\mathtt{z}_1][ \mathtt{P}_2/\mathtt{z}_2][ \mathtt{P}_3/\mathtt{z}_3] , \Idx)$
         can be immediately concluded. 
\item Rule $(P_6)$. Let ${\tt M} = \TTY  $ and
         $\kappa=\Pi \mathtt{y}^{(\sigma\rightarrow\sigma)}.\sigma$ such that $\sigma=\Pi\mathtt{z}_1^{\tau_1}. ... \Pi\mathtt{z}_m^{\tau_m}. \gamma$ where 
         $\gamma \in\{\Nat, \CIRC{\mathtt{E}}\}$.
         \begin{itemize}
         \item  The proof of $\gamma=\Nat$   is similar to the proof of the rule $(P_3)$. By  Lemma~\ref{compRephrased},
           we have to prove that $B=\{{\tt x}_1:\nu_1,...,{\tt x}_n:\nu_n\}$ implies $\vdash (\TTY \vv{[\mathtt{N}/\mathtt{x}]_1^{n}}) \mathtt{Q} \mathtt{P}_1 \ldots \mathtt{P}_m  : \Nat$,
  for all ${\tt N}_j$  such that $Comp(\emptyset,{\tt N}_j,\nu_j\vv{[\mathtt{N}/\mathtt{x}]_1^{j-1}})$ where $j\leq n$,
  for ${\tt Q}$  such that $Comp(\emptyset,{\tt Q}, (\sigma\rightarrow\sigma)\vv{[\mathtt{N}/\mathtt{x}]_1^{n}} )$,
  for all ${\tt P}_i$  such that $Comp(\emptyset,{\tt P}_i,\tau_i\vv{[\mathtt{N}/\mathtt{x}]_1^{n}}[\mathtt{Q}/\mathtt{y}]\vv{[\mathtt{P}/\mathtt{z}]_1^{i-1}} )$ where $i\leq m$.
  Namely, we have to prove that $\vdash  \TTY  \mathtt{Q} \mathtt{P}_1 \ldots \mathtt{P}_m : \Nat$ whenever $\vdash {\tt Q}:(\sigma\rightarrow\sigma)\vv{[\mathtt{N}/\mathtt{x}]_1^{n}}$
  and $\vdash {\tt P}_i:\tau_i\vv{[\mathtt{N}/\mathtt{x}]_1^{n}}[\mathtt{Q}/\mathtt{y}]\vv{[\mathtt{P}/\mathtt{z}]_1^{i-1}} $.
  This is true  by rules $(P_2)$ and $(P_6)$. 
\item The proof of $\gamma=\CIRC{\mathtt{E}}$  is similar to the proof of the rule $(P'_5)$.
    By  Lemma~\ref{compRephrased},
  we have to prove that  $B=\{{\tt x}_1:\nu_1,...,{\tt x}_n:\nu_n\}$ implies  
  $$\vdash (\TTY\vv{[\mathtt{N}/\mathtt{x}]_1^{n}})  \mathtt{Q} \mathtt{P}_1 \ldots \mathtt{P}_m
   : \CIRC{\mathtt{E}\vv{[\mathtt{N}/\mathtt{x}]_1^{n}} [\mathtt{Q}/\mathtt{y}]\vv{[\mathtt{P}/\mathtt{z}]_1^{m}}  }  \qquad \text{and}$$
   $Comp(\emptyset,{\tt E} \vv{[\mathtt{N}/\mathtt{x}]_1^{n}} [\mathtt{Q}/\mathtt{y}]\vv{[\mathtt{P}/\mathtt{z}]_1^{m}}, \Idx)$
 for all ${\tt N}_j$  s.t. $Comp(\emptyset,{\tt N}_j,\nu_j\vv{[\mathtt{N}/\mathtt{x}]_1^{j-1}})$ where $j\leq n$,
  for ${\tt Q}$  such that $Comp(\emptyset,{\tt Q}, (\sigma\rightarrow\sigma)\vv{[\mathtt{N}/\mathtt{x}]_1^{n}} )$,
  for all ${\tt P}_i$  such that $Comp(\emptyset,{\tt P}_i,\tau_i\vv{[\mathtt{N}/\mathtt{x}]_1^{n}}[\mathtt{Q}/\mathtt{y}]\vv{[\mathtt{P}/\mathtt{z}]_1^{i-1}} )$ where $i\leq m$.
The typing requirement is similar to that of the previous case.
Moreover, it is easy to check that $\wfi{B,\codom(B)\cup\{\sigma  \}}$ requires that $B\vdash \texttt{E}:\Idx$ is in the premises of $(P_6)$. 
Thus $Comp(B,  \texttt{E}, \Idx)$ follows by hypothesis,
and  $Comp(\emptyset,{\tt E}  \vv{[\mathtt{N}/\mathtt{x}]_1^{n}} [\mathtt{Q}/\mathtt{y}]\vv{[\mathtt{P}/\mathtt{z}]_1^{m}}, \Idx)$
         can be concluded.  
              \end{itemize}
\item Rule $(I_0)$. The proof follows by induction, it suffices to apply $(I_0)$ to obtain the typing in the induction hypothesis.
\item Rule $(I_1)$. Immediate by the evaluation rule $(n)$.
\item Rule $(I_2)$. The proof follows by induction, because we assume that $ \odot $ is a (generic) total operator.     
\item Rule $(I_3)$. The proof follows by induction and by    the evaluation rule $(sz)$.      
\item Rule $(C_1)$. Immediately $B\vdash {\tt U}:\CIRC{\num{k}}$, $B\vdash\num{k}:\Idx$ and $\num{k}\ev^1\num{k}$ hold, thus  $Comp(B, {\tt U},\CIRC{\num{k}})$.
\item Rules $(C_2),(C_3),(C_4),(C_5)$. The proofs are similar to that  of the rule $(P'_5)$.
\item Rule $(M)$. The proof is similar to that  of the rule $(P_5)$.\qed
 \end{itemize}
 \end{proof}

 $Comp$ has been defined in order to obtain the next corollary that states the strong normalization of closed term typed $\Nat$
 and the that its evaluation provides a unique result.

 \begin{corollary}[$\Idx$-normalization]\label{idxNormalization}
   Let $\vdash \mathtt{E} : \Idx$.
   \begin{enumerate}
   \item ${\tt E}\ev^1 \num{n}$, for some $\num{n}$.
     \item If $\CIRC{\mathtt{E}}$ occurs in a valid type derivation then ${\tt E}\ev^1 \num{n}$, for some $\num{n}$.
   \end{enumerate}
 \end{corollary}
 \begin{proof}
   \begin{enumerate}
   \item  $\vdash \mathtt{E} : \Idx$ implies $Comp(\emptyset, {\tt E},\Idx)$  by Theorem~\ref{thComp}; so  ${\tt E}\ev^1 \num{n}$ by Definition \ref{compDefinition}.
   \item Let $\mathcal{D}$ be a valid type derivation. By induction on $\mathcal{D}$ we can prove that if $\CIRC{\mathtt{E}'}$ occurs in the conclusion of the derivation
     then  $B\vdash \mathtt{E}' : \Idx$ is required for some $B$. Thus, since we assumed  $\vdash \mathtt{E} : \Idx$, we conclude by  the previous point.
     \qed
 \end{enumerate}

 \end{proof}

We can now focus on  standard programming properties (see \cite{pierce2002mit}).
A first property of a paradigmatic programming language as $\qpcf$ is \emph{preservation}, i.e. if a well-typed term takes a step of evaluation then the resulting term is also well typed. 
A second property expected for a programming language is \emph{progress}: well-typed terms evaluation does not get stuck.
Roughly,  a term $\tt P$ gets stuck whenever the evaluation of  $\tt P$  ends in a normal form, which is not a natural number.

\begin{corollary}\label{idxPreservationProgress}
  \begin{description}
  \item[(Preservation)] $\vdash \mathtt{M} : \Idx$ and $ \mathtt{M}\evaluates^1 \mathtt{N}$ then $ \vdash \mathtt{N}: \Idx$.
\item[(Progress)] $\vdash \mathtt{M} : \Idx$ and $ \mathtt{M}\evaluates^1 \mathtt{N}$ then $\mathtt{N}$ is a numeral $\num{n}$.
  \end{description}
\end{corollary}
\begin{proof}
  Immediate, by Theorem \ref{thComp}. \qed
\end{proof}

For remaining ground types we prove preservation and progress together. 

\begin{theorem}
  \label{th:preservation}
  \begin{enumerate}
\item  $\vdash \mathtt{M} : \Nat$ and $ \mathtt{M}\evaluates \mathtt{N}$   then  $ \vdash \mathtt{N}: \Nat$ and $\mathtt{N}$ is a numeral.
\item If $\vdash \texttt{M}:\CIRC{\mathtt{E}}$ and $\mathtt{M} \ev  \mathtt{N}$ then
  $\vdash \texttt{N}:\CIRC{\mathtt{E}}$ where $\mathtt{N}$ is a circuit; moreover,  $\mathtt{E}\ev\num{m}$ for some $\num{m}$.
  \end{enumerate}
\end{theorem}
\begin{proof}
  We recall that $\mathtt{M} \ev  \mathtt{N}$ means that there exists a derivation $\mathcal{D}$ concluding $\mathtt{M} \ev^{\PRB{\alpha}}  \mathtt{V}$ where $\alpha >1$.
  The statements are proved by  induction on the derivation proving $\mathtt{M}\evaluates \mathtt{N}$ by considering the typing hypothesis.
  \begin{enumerate}
  \item
    The proof is done by  induction on the derivation proving $\mathtt{M}\evaluates \mathtt{N}$.
$(n)$ is trivial. If the last applied evaluation rule is one between  $(s)$, $(p)$,
 $(\beta)$, $(\mathrm{if}_l)$, $(\mathrm{if}_r)$, $(Y)$, $(gt)$, $(st)$ then the proof follows by induction.
 We remark that the last rules can be $(sz)$ or $(op)$  because its result can be typed $\Nat$ via the typing rule $(x_3)$:
luckily, in both cases the proof is still immediate by Theorem \ref{thComp}. The cases $(u)$, $(u')$, $(u'')$, $(r_0)$, $(r_1)$,  $(r_2)$,  $(it)$  are not possible,
 because the typing hypothesis in the statement excludes them. The case $(m)$ easily follows by induction and the definition of $\text{circuitEval}$.
\item  $\mathtt{E}\ev\num{m}$ follows by Corollary~\ref{idxPreservationProgress}. The proof is similar to that of Lemma~\ref{lemma:circExp}.  
  \qed
 \end{enumerate}
\end{proof}


\section{Examples}\label{sec:examples}

In this section we propose some examples of  quantum circuit families implementing interesting algorithms.
As previously done, given a sequence $(b_1,\ldots,b_k)$ of bits, we write $\mathbf{n}(b_1,\ldots,b_k)$ to denote the numeral that represents it. 




\begin{example}[Deutsch-Jozsa]\label{ex:deutsch}
  We aim to program the Deutsch's algorithm~\cite{NieCh10} in $\qpcf$. This can be done by a term that  represents the entire (infinite) quantum circuit family.
  
The ``basic case'' of Deutsch's problem can be formulated as follows.
Given a block box ${B}_f$ implementing some  function $f : \{0,1\} \fleche \{0,1\}$, determine whether $f$ is constant or balanced. 
The classical computation to determine whether f is constant or balanced is very simple: one computes $f(0)$ and $f(1)$, and then check if $f(0) = f(1)$. This requires two different calls to ${B}_f$.
Deutsch showed how to achieve this result with a single call to ${B}_f$:  Deutsch's algorithm exploits \emph{quantum parallelism} phenomenon
\commento{. 
Over-simplifying, quantum parallelism allows to a quantum computer }{ that allows }to evaluate a function $f(x)$ for different values $x$ at the same time.
\commento{ and 
it is able to extract  information about more than one of the values $f(x)$ from a superposition state.
Each function $f:\{0,1\}\redto\{0,1\}$ can be represented by a circuit $\ket{x, y\oplus f(x)}$, where the operator $\oplus$ represents the addition modulo 2. The transformation $\ket{x, y}\redto\ket{x, y\oplus f(x)}$ is unitary. 
In particular, taking $x=\frac{\ket{0}+\ket{1}}{\sqrt{2}}$ and $y=0$ as inputs, the output state will be $\frac{\ket{0, f(0)}+\ket{1, f(1)}}{\sqrt{2}}$, which  provides information about $f(0)$ and $f(1)$ simultaneously.}{
}

The problem can be generalized considering a function $f : \{0,1\}^{n} \fleche \{0,1\}$ which acts on many input bits. This yields  the n-bit generalization of Deutsch's algorithm, known as the Deutsch-Josza algorithm. 
The following picture represents the circuit, up to the last, measurement phase{}.

     \begin{center}
       \scalebox{0.50}
       {
         \ifx\pdfoutput\undefined 
         \epsfbox{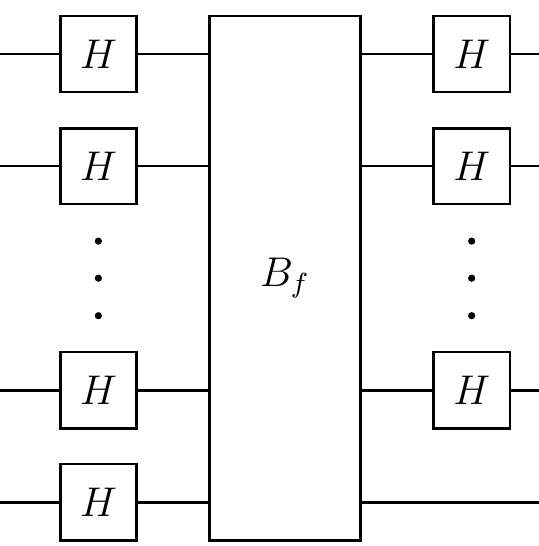} \else 
         \includegraphics{DJmeno} \fi
       }
     \end{center}
When fed with a classical input state of the form $\ket{0\ldots 01}$, the output measurement of the first $n-1$ bits reveals if the function $f$ is constant or not. If all $n-1$ measurement results are 0, we can conclude that the function was constant. Otherwise, if at least one of the measurement outcomes is 1, we conclude that the function was balanced. See~\cite{NieCh00} for { more details.}

Let $\mathtt{H}:\CIRC{\num{0}}$ and $\mathtt{I}:\CIRC{\num{0}}$  be the (unary) Hadamard and Identity gates respectively.
Suppose $\mathsf{M}^{B_{f}}:\CIRC{\num{n}}$  is given for some $n$ such that $\mathsf{M}^{B_{f}} \ev \mathsf{U}_{f}$ where $\mathsf{U}_{f}:\CIRC{\num{n}}$ is the $\qpcf$-circuit that represents 
the the black-box function $f$ having arity  $n+1$. 

Observe that
$\lambda \mathtt{x}^{\Idx}.\iterCirc\;  \mathtt{x}  \mathtt{H}  \mathtt{H}: \Pi \mathtt{x}^{\Idx}.\CIRC{\mathtt{x}}$  generates $x+1$ parallel copies of Hadamard gates $ \mathtt{H}$, and 
$\lambda  \mathtt{x}^{\Idx}.\iterCirc\; \mathtt{x}  \mathtt{I}  \mathtt{H} :\Pi  \mathtt{x}^{\Idx}.\CIRC{ \mathtt{x}}$
concatenates in parallel $x$ copies  of  Hadamard gates $ \mathtt{H}$ and one copy of the  identity gate $\mathtt{I}$.
Thus the parametric measurement-free Deutsch-Jozsa circuit  
can be  defined as
$$\mathsf{DJ^{-}}=\lambda \mathtt{x}^{\Idx}.\lambda \mathtt{y}^{\CIRC{\mathtt{x}}}.
  ((\iterCirc\, \mathtt{x}\, \mathtt{H}\, \mathtt{H})\appCirc \mathtt{y}) \appCirc(\iterCirc\,\mathtt{x}\, \mathtt{I}\, \mathtt{H} ):\sigma$$ 
where
 $\sigma=\Pi\mathtt{ x}^\Idx.\CIRC{\mathtt{x}}\fleche\CIRC{\mathtt{x}}$.

The last phase is performed by the operator  $\dMeas$, suitably fed with the representation of the classical input state, i.e. $\mathbf{n}(\underbrace{0\ldots0}_{n}1)$.
\commento{The evaluation yields derivations ending in
 $\dMeas(\MZ{}{\mathbf{n}(\underbrace{0\ldots0}_{n}1)},\mathsf{DJ^{-}}\num{n}\mathsf{M}^{B_{f}})\ev^{\PRB{\alpha}} \num{m}_\alpha$
 if and only if $\num{m}_\alpha$ is a  result (with probability $\alpha$).
}{
The evaluation yields
 $\dMeas(\mathbf{n}(\underbrace{0\ldots0}_{n}1),\mathsf{DJ^{-}}\num{n}\mathsf{M}^{B_{f}})\ev^{\PRB{1}} \num{m}$
 where $\num{m}$ is the result (with probability $1$).
 }

It is straightforward to make parametric the above term.
It suffices to replace $\num{n}$ with the  variable ${\tt n}^\Idx$, to replace the black-box with a variable $ {\tt b}^{\CIRC{\mathtt{n}}} $
so that, the resulting term is typed  $\Pi{\tt n}^\Idx.\Pi{\tt b}^{\CIRC{\mathtt{n}}}. \Nat $, or more simply $\Pi{\tt n}^\Idx.\CIRC{\mathtt{n}}\fleche \Nat $.
\qed
\end{example}

\begin{example}[Grover's searching algorithm]\label{ex:grover}
We provide the \qpcf\  encoding of the circuit that implements Grover's searching algorithm~\cite{NieCh00}.
More precisely, Grover's algorithm solves the problem of a search in a given a set $X = \{x_1,x_2,\ldots,x_{N}\}$ of $N=2^{n+1}$ elements.
Given a boolean function $f : X \rightarrow \{0,1\}$, the target is to find an element $x^*$ in $X$ such that $f(x^*) = 1$.
 With a classical circuits, one cannot do better than performing a linear number of queries to find the target element. 
Grover's quantum solves algorithm the  search   in $O(\sqrt{N})$.
The main idea of  Grover's searching algorithm is to make a fair superposition of input elements 
and, then, to iterate $O(\sqrt{N})$ time a subroutine applying
 an ``oracle gate'' $\textsf{O}^*$ that encodes the function $f$ and a suitable diffusion operator $\textsf{D}$. 
 After each application of the subroutine, it is possible to show that
{the probability to measure the target element $x^*$ increases: the amplitudes $\alpha^*$  goes up}
   by more than $\frac{1}{\sqrt{N}}$. 
This meas that, after $O(\sqrt{N})$ repetitions,  $\alpha^*$ is very close to $1$, and thus a final, total measurement, yields $x^*$ with a negligible error.

Some slightly different circuit implementations have been proposed in literature~\cite{KLM07,NO08}, we follow~\cite{NieCh00}:
\begin{align*}
 \Qcircuit @C=1em @R=.7em {
                   &         &                      &                         &                      & \ustick{\textsf{D}} \\
  \lstick{\ket{0}} & /^n \qw & \gate{H^{\otimes n}} & \multigate{1}{\textsf{O}^*} & \gate{H^{\otimes n}} & \gate{2 \ket{0^n}\bra{0^n} - I_n}         & \gate{H^{\otimes n}} & \qw & /^n\qw & \cdots &  \\ 
  \lstick{\ket{1}} & \qw     & \gate{H}             & \ghost{\textsf{O}^*}        & \qw                  & \qw                                       & \qw                & \qw  & \qw & \cdots & \\
                   &         &                      &                         &                      & \dstick{\text{ $O(\sqrt{N})$ times}}
  \gategroup{2}{5}{2}{7}{.7em}{^\}}
  \gategroup{2}{4}{3}{10}{.7em}{_\}}
 }
\end{align*}\\[5mm]

\noindent Let $\mathsf{O}^{*}$ be the  circuit that represents  the function mapping $\ket{x,q}$ to $\ket{x,q\oplus f(x)}$. 
Look at the figure above: the lower wire (initialized to $\ket{1}$) feeds $\mathsf{O}^{*}$ with the state $(\ket{0}-\ket{1})/\sqrt{2}$ 
while the other ones (initialized to  ${\underbrace{\ket{0\ldots 0}}_n)}$ feed $\mathsf{O}^{*}$ with the fair superposition of all possible inputs.
Therefore, $\mathsf{O}^{*}$ maps $\ket{x}\ket{\frac{\ket{0}-\ket{1}}{\sqrt{2}}}$ on $(-1)^{f(x)} \ket{x}\ket{\frac{\ket{0}-\ket{1}}{\sqrt{2}}}$;
if we neglect the (bottom) ancillary wire then $\ket{x}$ is mapped on $(-1)^{f(x)} \ket{x}$ producing a key flipping behavior.

Let $\mathsf{M}^{O^*}:\CIRC{\num{n}}$ be a \qpcf\ term such that $\mathsf{M}^{O^*} \ev^{\PRB{1}} \mathsf{O}^{*}$ {}{(we recall that $\CIRC{\num{n}}$ is the type for $n+1$ wires.) }
Let $\mathsf{M}^{D}:\CIRC{\mathtt{x}}$ be such that $(\lambda \mathtt{x}^\Idx.\mathsf{M}^{D})(\Pred\, \num{n}) \ev \mathsf{D}$ where $\mathsf{D}$ is the circuit that represents the diffusion operator 
that magnifies amplitudes:  $\textsf{D}$
 maps the superposition $\sum_{i=0}^{N-1}\alpha_iV_i$ in $\sum_{i=0}^{N-1}(2m -\alpha_i)V_i$ where $m=(\frac{1}{N}) \sum_{i=0}^{N-1}\alpha_i$ is the average of all amplitudes
(see \cite{shor02psam}). Some interesting decomposition of $D$ in terms of smaller quantum circuits can be found in \cite{diao2002,glos2016qip}.

Let ${\sqrt{\phantom{x}}}:\Nat\rightarrow\Nat$ be a $\qpcf$-term calculating (integer approximation of) the square root and 
let $\mathsf{M^{x}_{seq}}: \CIRC{\mathtt{ \mathtt{x}}}\rightarrow\Nat\rightarrow\CIRC{\mathtt{x}}$ be $\mathsf{M^{A}_{seq}}\,\mathtt{x}$, i.e.
the application of the term defined in Example \ref{ex:seqcirc} to  $\mathtt{x}$.
The term that implements the core of Grover's algorithm  is:
$$\mathsf{M_{G}}=\lambda \mathtt{x}^{\Idx}.\mathtt{z}^{\CIRC{\mathtt{x}}}.(\iterCirc\,\mathtt{x}\, \mathtt{H}\, \mathtt{H})
\appCirc(\mathsf{M^x_{seq}}(\mathtt{z}\appCirc (\mathsf{M}^{D}\parallel \texttt{I})) \, \sqrt{\mathtt{ x}}):\sigma$$
where $\sigma=\Pi\mathtt{ x}^\Idx.\CIRC{\mathtt{x}}\fleche \CIRC{\mathtt{x}}$.
A measurement of $\mathsf{M_{G}}$ applied to a suitable arity $\num{n}$ and a suitable oracle operator  $\mathsf{O}^{*}$
  allows  to execute  Grover's algorithm. 

 \medskip
 For example, 
 we can consider a search in a space of $2^{3}=8$ states looking for  $\ket{011}$.
 Grover's circuit take as input the (classical) state \ket{0001}. 
 
 It is  easy to verify that   $\mathsf{G}_3=\dMeas(\mathbf{n}(0001), \mathsf{M_{G}}(\num{3})(\mathsf{M}^{\mathsf{O}^{*}}))$ solves the search problem with a bounded error.
 In particular,  $\mathbf{D}(\mathsf{G}_3, \mathbf{n}(0110))$ is a set including just one derivation, namely the derivation concluding
 $\mathsf{G}_3\ev^{\PRB{\alpha}}\mathbf{n}(0110)$ such that $\alpha= 0,945$.
 Thus,  the evaluation of $\mathsf{G}_3$ gives the right results with a  94,5\% of probability.
 See~\cite{lecturenotesstrubell} for the details.
\qed
\end{example}

\section{Related Work}\label{sec:related}



In this section, we sketch the state of the art of quantum programming languages by focusing on calculi related to $\qpcf$.

\medskip

{After a first formal attempt by Maymin~\cite{May97},}  Selinger rigorously defined
a first-order quantum functional language~\cite{selinger2004mscs}.  Subsequently the author, in a joint work with Valiron~\cite{selinger2006mscs}, 
defined a quantum $\lambda$-calculus, (that we dub $\lambda_{sv}$ in what follows),  with classical control and explicitly  based on the QRAM architecture~\cite{Knill96}.
$\lambda_{sv}$ rests on unitary transformations on quantum states and an explicit measurement operator which allows the program to observe the value of one quantum bit.
The separation between data and control is explicit.
The type system of $\lambda_{sv}$  avoids run time errors, enjoys good properties such as subject reduction, progress, and {safety} and is based on the affine intuitionistic Linear Logic.
This permits a fine control over the linearity of the system, by distinguishing between duplicable and non-duplicable resources.
{$\lambda_{sv}$ can be seen as the departing point of several investigations.  On the foundational side, see}
for instance~\cite{dallago2009mscs,DLMZ10tcs}. 
  On the semantic side, we just cite \cite{selinger09chap,HH11,pagani14acm}.

{ $\qpcf$ follows the slogan quantum data\&classical control (``$qd\&cc$'') proposed for $\lambda_{sv}$ in~\cite{selinger2006mscs},
  albeit its typing system does not explicitly include linear/exponential  types.
}

\medskip
Looking for implementation-oriented proposal, the most interesting reality is Quipper, an embedded, scalable functional quantum programming language. Quipper is essentially a  circuit description language: circuits can be created, manipulated, evaluated  in a mixture of procedural and declarative programming styles. The most important quantum algorithms can be easily encoded thanks to a number of programming tools, macros,  and extensive libraries of quantum functions. 
   
Quipper is {} based on the lambda calculus with classical control proposed in~\cite{selinger2006mscs}, and this relationship has been partially explained in~\cite{Ross2015} by means of the calculus Proto-Quipper.  
Reduction rules are defined between configurations as in~\cite{selinger2006mscs,Zorzi16} and, since the calculus is measurement free, they are totally deterministic (likewise to~\cite{dallago2009mscs,DLMZ10tcs,Zorzi16}).
Proto-Quipper type system  
 is based on
intuitionistic linear logic (with both additive and multiplicative modalities) plus {a} type for circuits. {A more recent} version of Proto-Quipper, called Proto-QuipperM, has been defined by Selinger and Rios in~\cite{rios2017eptcs} {together with an interesting categorical semantics.}

{
  $\qpcf$ follows the trend started by Quipper about the management of quantum circuits as prominent classical data.
  However,  it differs from Proto-Quipper formalization, at least,  for the absence of quantum states and linear/exponential types,
but also for the presence of dependent types.}
\medskip 

{Another} recent quantum language is  \textsf{QWire} introduced in~\cite{qwire}. 
Also \textsf{QWire} can be seen as a language for circuit manipulation.
It rests on the QRAM model and
{} aims at separating classical and quantum part of the computation. 
\textsf{QWire} {} is a very simple and manageable linear language for the definition of quantum circuit {that,} through a sophisticate interface, 
can be treated as  a ``quantum plugin'' for an host classical language. This is reflected by the type system, inspired to  Benton's LNL Logic
that partitions the exponential data into a purely linear fragment and a purely non-linear fragment connected via a categorical adjunction.
{Circuits} are treated as classical data in the host language through a clever interface based on a box-unbox mechanism. 
Moreover, the authors show that \textsf{QWire}  is able to deal with a depend-type based host language via an interesting example. 
Albeit they have been developed independently, \textsf{QWire} and $\qpcf$ have a many aspects in common, and both move the focus from states to circuits. 
Thus, \textsf{QWire} deserves a direct comparison with $\qpcf$.\\
$\qpcf$ bans linear typing, while \textsf{QWire} use it to provide a very helpful support for a quantum circuit language description 
that borrows the best features from the ``hardware circuit description'' languages (Verilog, VHDL, ...).
A future extension of $\qpcf$ should include  a  circuit language inspired  to  \textsf{QWire}.\\
Second,  \textsf{QWire}
{is a plug-in quantum extension of a classic language, and its
}  dependent types {are} made available from the host language, while $\qpcf$ is a stand-alone language that includes a limited form of dependent types.
Finally, \textsf{QWire} allows partial measurements, thus it supports quantum states that mutually interact with the classical environment in the frame of the general QRAM model;
in $\qpcf$ we focused on {restricted} co-processors, as widely explained in Section \ref{sec:qram}.
{Summarizing,  \textsf{QWire} and $\qpcf$ shares many aspects,
  but { \textsf{QWire} provides more programming flexibility in the implementation of quantum algorithms than $\qpcf$, while $\qpcf$ is a standalone language with cheaper hardware requirements than  \textsf{QWire}. }

\medskip
{
  For the sake of completeness, we remark that other approches to quantum programming languages  exist in literature.
  Interesting proposal are, among the others: the functional quantum language  (based on strict Linear Logic) QML~\cite{AltGra05,AltGra05bis,Grattage11entcs};
  the  linear algebraic $\lambda$-calculi~\cite{ArrDow08,Vaux06,Vaux09,ArrighiDiazcaroLMCS12};
 the \emph{measurement calculus} ~\cite{Nie03,DKP07} developed as an efficient rewriting system for measurement based quantum computation;
and, last, the quantum-control quantum-data paradigm described in \cite{minsheng16}.
}


 \section{Conclusions}\label{sec:conclusions}\label{sec:discussion}

\subsection{On the expressive power of \qpcf}\label{sec:expressivepower}

\qpcf\ is a language able to describe parametric quantum circuit families, in accord with \cite{rios2017eptcs}, where the difference between the notions of \emph{parameter} (informally, the information available at the compile time, as e.g. the input dimension) and \emph{state} (the information available at run time, e.g. the effective value of quantum data) is highlighted. 
{
  Indeed,  Examples \ref{ex:deutsch} and \ref{ex:grover} show how to use dependent types to represent circuit families.
  The study of the exact expressive power of \qpcf, in terms of a formal notion of uniform quantum families is left to future work.
  Since the limitations we assumed on expressions of type $\Idx$, it is clear that many representations are unsuitable.
}

A related interesting point is the expressiveness of $\qpcf$ w.r.t. quantum algorithms. 
Since we restrict measurements to total, deferred ones, in  $\qpcf$ one {cannot directly represent algorithms that exploit partial measurement during the computation} (see, for example, Shor's original formulation of the factorization algorithm~\cite{Shor94}).

Of course, it is possible to encode their ``deferred versions'' exploiting the deferred measurement principle. See~\cite{NieCh10} (Section 4.4) for a careful account about the rewriting of a circuit that allows intermediate partial measurement into an equivalent deferred form.


Finally, we remark that by endowin  \qpcf\ with an \emph{universal basis} of quantum gates ensures the full expressivity w.r.t. quantum transformations, up to  a  definition of gate approximation (see \cite{NishOz09} for details).
{We also argue that particular choices of gates in \qpcf\ could return interesting instances of the language.
  For instance, all reversible circuits are a subset of quantum ones.
  See~\cite{PaoliniPiccoloRoversi-ENTCS2016,paolini2018ngc} for a recent characterization of the reversible computing.
  In particular,  $\qpcf$ appears to be a simple setting where reversible and classical computation coexists and, potentially, can cooperate.
}

\subsection{Conclusive Statement}

We study $\qpcf$, an extension of $\PCF$ for quantum circuit generation and evaluation. 
$\qpcf$ pursues seriously the $qd\&cc$ paradigm
  in a restricted QRAM environment where only total measurements are allowed. First, this
makes quantum programming easy: we can program  circuit descriptions by using only classical data.
{Second, this approach is cheaper than the usual on hardware requirements.}
 
In this work, we {explain} $\qpcf$ syntax, typing rules and a possible formulation of the evaluation semantics. We prove some basic properties of the language. 
We provide some encoding examples of parametric circuit families  that exploit the expressive power of \qpcf. 

The careful analysis of the exact expressive power of $\qpcf$ w.r.t. {formal notion of circuit families} is an open question we are currently  addressing (following~\cite{dallago2009mscs}, where a two-way correspondence between a formal calculus and the finitely generated quantum circuit families~\cite{NishOz09} is proved).
Finally, even if the use of total measurement does not represent a theoretical limitation, a partial measurement operator can represent a useful programming tool. Therefore, an interesting task will be to integrate in $\qpcf$ the possibility to perform partial measures of computation results.


\bibliography{biblio}	
\bibliographystyle{abbrv}

\end{document}